\documentclass[onecolumn,prx,nofootinbib, superscriptaddress,notitlepage]{revtex4-1}

\usepackage[colorlinks=true,linkcolor=blue,urlcolor=blue,citecolor=blue,anchorcolor=green,pdfusetitle]{hyperref}
\usepackage{amsthm}
\usepackage{amsmath}
\usepackage{appendix}
\usepackage{bbold}

\usepackage{tikz}
\usepackage{pgfplots}
\usetikzlibrary{patterns}
\usepackage{float}
\usepackage{graphicx}
\usepackage{nicefrac}
\usepackage{stmaryrd}

\usepackage{xcolor}
\definecolor{orange}{HTML}{FF8E00}
\definecolor{lblue}{HTML}{e0f3f8}
\definecolor{mblue}{HTML}{85A0DD}
\definecolor{dblue}{HTML}{264A9C} 
\definecolor{red}{HTML}{FF4100}
\definecolor{yellow}{HTML}{FFBC00}
\definecolor{gray}{HTML}{8c8c8c}
\definecolor{purple}{HTML}{990099}

\usepackage{algorithm}
\floatname{algorithm}{Protocol}
\usepackage{algcompatible}
\usepackage{caption}
\usepackage{subcaption}
\usepackage{enumerate}

 \theoremstyle{plain}
  \newtheorem{thm}{Theorem}
  
  \newtheorem{lemma}[thm]{Lemma}
  \newtheorem{cor}[thm]{Corollary}
  
  \newtheorem*{thm*}{Theorem}
  \newtheorem*{prop*}{Proposition}
  \newtheorem*{lemma*}{Lemma}
  \newtheorem*{cor*}{Corollary}
  \newtheorem*{remark*}{Remark}

\theoremstyle{definition}
\newtheorem{defn}[thm]{Definition}
 \newtheorem*{conj*}{Conjecture}
 \newtheorem{conj}[thm]{Conjecture}

\usepackage{todonotes}

 \newcommand{\tr}{\mathrm{Tr}}

 \newcommand{\id}{\mathbb{I}}
 \DeclareMathOperator{\idmap}{id}
 
 \newcommand{\clas}[1]{\mathsf{#1}}
 
 \newcommand{\ket}[1]{|#1\rangle}
 
  \newcommand{\proj}[1]{|#1\rangle\langle#1|}

  \newcommand{\cX}{\mathcal{X}}
  \newcommand{\cY}{\mathcal{Y}}
  \newcommand{\cH}{\mathcal{H}}
  \newcommand{\eps}{\varepsilon}
 
\begin{document}

\title{Upper bounds on device-independent quantum key distribution rates\\[.25em] and a revised Peres conjecture}

\author{Rotem Arnon-Friedman} 
\affiliation{EECS Department, University of California Berkeley, Berkeley, California 94720, USA}
\author{Felix Leditzky}
\affiliation{Institute for Quantum Computing, University of Waterloo, Waterloo, Ontario N2L 3G1, Canada}
\affiliation{Department of Combinatorics and Optimization, University of Waterloo, Waterloo, Ontario N2L 3G1, Canada}
\affiliation{Perimeter Institute for Theoretical Physics, Waterloo, Ontario N2L 2Y5, Canada}
\affiliation{JILA \& Center for Theory of Quantum Matter, University of Colorado Boulder, Boulder, Colorado 80309, USA}

\date{\today}

\begin{abstract}
Device-independent quantum key distribution (DIQKD) is one of the most challenging tasks in quantum cryptography. 
The protocols and their security are based on the existence of Bell inequalities and the ability to violate them by measuring entangled states. 
We study the entanglement needed for DIQKD protocols in two different ways. 
Our first contribution is the derivation of \emph{upper bounds} on the key rates of CHSH-based DIQKD protocols in terms of the violation of the inequality; this sets an upper limit on the possible DI key extraction rate from states with a given violation. Our upper bound improves on the previously known bound of Kaur et al.
Our second contribution is the initiation of the study of the role of \emph{bound entangled} states in DIQKD. We present a revised Peres conjecture stating that such states cannot be used as a resource for DIQKD. We give a first piece of evidence for the conjecture by showing that the bound entangled state found by Vertesi and Brunner, even though it can certify DI randomness, cannot be used to produce a key using protocols analogous to the well-studied CHSH-based DIQKD protocol.
\end{abstract}

\maketitle

\section{Introduction}

	The goal of a quantum key distribution (QKD) protocol is to allow two honest and cooperating parties, named Alice and Bob, to create a shared key unknown to anyone but them. To this end they hold a physical apparatus that, ideally, performs some measurements on quantum states, as defined by the protocol. In practice, however, the devices may perform different operations, due to either imperfections (such as noise in the quantum channels) or malicious intentions of the provider of the devices who wishes to gain information about the shared key. 
	\emph{Device-independent} (DI) QKD protocols address this problem; they achieve the strongest form of security of QKD protocols, with guarantees that hold even when the physical devices used to implement the protocol are faulty or produced by an untrusted manufacturer~\cite{ekert2014ultimate}. 
	
	To achieve the required level of security while not relying on the inner-workings of the physical devices, DI protocols are based on the violation of a Bell inequality~\cite{brunner2014bell, scarani2019bell}. 
	A Bell inequality~\cite{bell1964einstein} can be viewed as a ``non-local'' game played between Alice and Bob. Alice and Bob each hold a device in their own separated laboratories.   
	The questions of the game are the inputs Alice and Bob give to their devices (e.g., by pressing a button on the device): one question for Alice's device, described as a random variable~$\clas{X}$, and one for Bob's device with random variable $\clas{Y}$. 
	Once the inputs are given, the devices perform some (unknown) operation and return an output (displayed on a screen of the device, for example).
	The outputs produced by the devices are the answers of the game, corresponding to a random variable $\clas{A}$ for Alice and~$\clas{B}$ for Bob. The game is won if the questions and answers~$\left(\clas{x},\clas{y},\clas{a},\clas{b}\right)$ satisfy some pre-defined winning condition. 
	The most well-known non-local game is the CHSH game~\cite{clauser1969proposed}, in which the inputs and outputs are all bits, i.e.,~$\clas{x},\clas{y},\clas{a},\clas{b}\in\{0,1\}$, and the winning condition is given by the relation~$\clas{a}\oplus\clas{b}=\clas{x}\cdot\clas{y}$ (namely, if this relation holds we say that Alice and Bob win).

	Some non-local games have a special property that makes them useful in the context of DI cryptography: there exist devices which perform measurements on \emph{entangled} states that can succeed in the game better than any classical, local devices.
	In particular, these classical devices do not use entanglement to play the game.
	In the case of the CHSH game, for example, devices that do not use entanglement can reach a winning probability of at most $75\%$, while a device measuring a maximally entangled state with certain measurements leads to a winning probability of approximately~$86\%$~\cite{clauser1969proposed}.
	Observing that the devices can be used to win the game with winning probability above the classical limit thus certifies that entanglement is being used and, furthermore, that \emph{randomness}, or \emph{entropy}, is being produced. If a sufficient amount of randomness is being produced then it can be used to create a shared key in DIQKD protocols~\cite{ekert1991quantum,mayers1998quantum}. 
	This idea lies at the core of all DI cryptographic protocols.

	As entanglement is a necessary resource for implementing DIQKD protocols, it is interesting to study the relation between entanglement and DIQKD protocols or, more precisely, the  \emph{key rates} of such protocols. 
	 Roughly speaking, the key rate of a protocol describes the length of a secret shared key that can be created by executing the protocol. 
	 We say that a quantum state acts as a ``resource'' for a DIQKD protocol if it can be used to execute the protocol, i.e., the device used in the protocol performs some measurements on that resource state.
	 When comparing protocols, a protocol that can extract more key from the given resource state is a better protocol. 
	 Relating the amount\footnote{The ``amount of entanglement'' can be quantified in different ways. In the scope of this work we consider the winning probability of the state in a non-local game (such as the CHSH game) as the property quantifying entanglement.} of entanglement of a resource state to the amount of extractable key using a DIQKD protocol tells us how useful that state is in this context and can lead to fundamental insights in the study of quantum entanglement.

	The most well-studied DIQKD protocol is based on the CHSH game and exploits one-way classical communication from Alice to Bob (further details below); let us call this protocol the ``standard protocol''.  
	The focus of the previous works so far was to \emph{lower}-bound the key rate, i.e., to find the \emph{minimal} length of a key that can be created, in terms of the relevant parameters of the protocol. 
	With improved lower bounds over the years~\cite{reichardt2013classical, vazirani2014fully, miller2016robust}, a \emph{tight} lower bound was recently provided~\cite{arnon2018practical,arnon2019simple}.\footnote{The tight key rates of~\cite{arnon2019simple} match those achieved by a previous work~\cite{pironio2009device} that considered a weaker form of security.} 
	These lower bounds can be seen as a statement regarding the \emph{minimal} amount of key that can be extracted from an entangled state, winning the CHSH game with a certain probability, using the standard DIQKD protocol.  
	
	The following question then arises: Given an entangled resource state winning the CHSH game (or some other non-local game), what is the \emph{maximal} amount of key that can be extracted from it? To maximize the key rate we can construct protocols that are potentially better than the standard protocol, in terms of the amount of key that can be extracted for the given source considered. 
	Answering the above question is challenging as one first needs to come up with new potential protocols and then prove their security with the hope of achieving better key rates. This is a daunting process that may last for years without significant improvement in terms of the achieved parameters. To date, only one potential protocol, employing two-way communication, was suggested and a partial security proof was given~\cite{tan2019advantage}.\footnote{An additional recent result~\cite{Schwonnek2020robust} presented a new protocol with higher rates compared to the standard protocol by changing the ``key generation'' rounds rather than the classical processing steps.}
	
	Instead of embarking on the quest for better protocols one may take the somewhat opposite direction and look for \emph{upper} bounds on the possible key rates of a given resource state. 
	More concretely, one can aim to show that \emph{no} alternative protocol can lead to key rates better than a certain upper bound. 
	This sheds light on the relation between entanglement and DIQKD protocols by setting limits on the 	usefulness of entangled states for DIQKD protocols.	
	
	Such upper bounds on key rates are also relevant from the more ``applicative'' point of view.
	Wishing to implement DIQKD protocols in experiments, we need to push the theory further so that the entangled resource states needed for producing a positive key rate can be created in realistic noisy experimental settings. 
	Investigating upper bounds allows us to see the limitations of a whole class of protocols all together, instead of looking for different protocols and analyzing them one by one. 
	
	Only a couple of recent works~\cite{kaur2020fundamental, winczewski2019upper} take this path and derive upper bounds on the key rates of the most general DIQKD protocols. Of special interest for us is the work by Kaur \textit{et al.}~\cite{kaur2020fundamental}. There, a new information theoretic quantity, termed \emph{intrinsic non-locality}, is introduced and shown to act as an upper bound on the key rates of DIQKD protocols in the presence of a quantum adversary.\footnote{This is in contrast to the work by Winczewski \textit{et al.}~\cite{winczewski2019upper}, where a super-quantum adversary is considered. Upper bounds on the key rates in the presence of such adversary do not imply upper bounds on the key rates in the presence of a quantum adversary, which is the more relevant scenario.}
	The intrinsic non-locality is defined via an optimization problem that for general correlations might be hard to solve exactly.
	For protocols based on the CHSH game, an explicit upper bound obtained by relaxing the optimization problem was derived in~\cite{kaur2020fundamental}.
	
	In this work, instead of working with the intrinsic non-locality, we consider a closely related information-theoretic quantity called the intrinsic information~\cite{christandl2007unifying}, and use it to derive an upper bound on the key rates of CHSH-based DIQKD protocols.
	Our bound is straightforward to derive and significantly improves upon the bound presented in~\cite{kaur2020fundamental}.
	In some sense, the improvement builds on insights related to the derivation of the \emph{lower} bounds on the key rates of the standard protocol. 
	This is explained in detail in Section~\ref{sec:CHSH_bounds}.
	  
	
	We then continue to study the relation between entanglement and DIQKD protocols by considering the usefulness of \emph{bound-entangled states} for such protocols. 
	In~1999, Peres~\cite{peres1999all} conjectured that Bell non-locality is equivalent to distillability of entanglement. Namely, bound-entangled states cannot be used to violate a Bell inequality or win a non-local game. 
	 After much effort the conjecture was disproven and a bound-entangled state was shown to violate a Bell inequality~\cite{vertesi2014disproving}. 
	 In fact, the work \cite{vertesi2014disproving} took one step further and showed that the violation of the considered Bell inequality can also be used to certify the production of randomness in the DI setting. That is, bound-entangled state are useful, in the sense that they can be used as a resource for DI randomness certification protocols.
	 A stronger version of the Peres conjecture, stating that bound-entangled states cannot be used to violate so-called steering inequalities, was also proven wrong~\cite{moroder2014steering}. 
	
	We put forward a natural revised conjecture: Bound-entangled states \emph{cannot} be used to produce a \emph{secret key} using DIQKD protocols. Our analysis of upper bounds for DIQKD protocols together with examinations of the bound-entangled state known to violate the Bell inequality considered in~\cite{vertesi2014disproving} seem to suggest that, in contrast to the original Peres Conjecture, the revised conjecture may hold.
	
	The manuscript is arranged as follows. 
	Section~\ref{sec:upper_bounds} deals with deriving upper bounds on the rates of DIQKD protocols: we first define the device-dependent and DI secret key capacities in Section~\ref{sec:capacity-definitions}. We then discuss the relevant information-theoretic quantities in Section~\ref{sec:intrinsic}, and present our explicit bounds for protocols based on the CHSH inequality in Section~\ref{sec:CHSH_bounds}. 
	Finally, we present and explain the ``revised Peres conjecture'' in Section~\ref{sec:peres}.

	 \paragraph*{Notation.} We use  $A, B, \dots$ to denote quantum registers while $\clas{A},\clas{B},\dots$ are used for classical registers or random variables (RV). We denote by lowercase letters a particular value of an RV, i.e., an RV $\clas{X}$ defined on an alphabet $\cX$ may take value $\clas{x}\in\cX$ with probability $p(\clas{x})$.
	 Quantum systems $A,B,\dots$ are associated with finite-dimensional Hilbert spaces $\cH_A, \cH_B,\dots$, and composite systems $AB\dots$ are described using the tensor product $\cH_A\otimes \cH_B\otimes \dots$. 
	 A (mixed) quantum state $\rho_A$ on a quantum register $A$ is a positive semidefinite linear operator of unit trace acting on $\cH_A$.
	 A pure state $\psi_A$ has rank $1$, and may be identified with a normalized vector $\ket{\psi}_A\in\cH_A$ satisfying $\psi_A = \proj{\psi}_A$.
	 A quantum channel $\Lambda\colon A\to B$ is a linear completely positive trace-preserving map between the algebras of linear operators on $\cH_A$ and $\cH_B$.
	 We denote by $\id_A$ the identity operator on $\cH_A$, and by $\idmap_A$ the identity map on the algebra of linear operators on $\cH_A$.
	 A positive operator-valued measure (POVM) on $\cH_A$ is a collection $\lbrace \Pi_\clas{x}\rbrace_{\clas{x}}$ of operators acting on $\cH_A$ satisfying $\Pi_\clas{x}\geq 0$ for all $\clas{x}$ and $\sum_{\clas{x}} \Pi_\clas{x} = \id_A$.
	 A classical RV $\clas{X}$ with probability function $p(\clas{x})$ can be regarded as a quantum state via the embedding $\rho_{\clas{X}}=\sum_{\clas{x}} p(\clas{x}) \proj{\clas{x}}_\clas{X}$ for some fixed orthonormal basis $\lbrace\ket{\clas{x}}_\clas{X}\rbrace_\clas{x}$ of $\cH_{\clas{X}}$.
	 A classical-quantum or cq state is a state of the form  $\rho_{\clas{X}A} = \sum_{\clas{x}} p(\clas{x}) \proj{\clas{x}}_\clas{X} \otimes \rho_A^{\clas{x}}$, where $\clas{X}$ is an RV with probability function $p(\clas{x})$, and $\rho_A^{\clas{x}}$ is a quantum state on $A$ for all $\clas{x}$.
	 This definition generalizes to more than two parties (i.e., ccq states) in the obvious way.
	 
	 We also define the following entropic quantities for (quantum, classical, or hybrid) states $\rho_{ABC}$ and its marginals:
	 \begin{itemize}
	 	\item the von Neumann entropy $H(A)_\rho = -\tr(\rho_A\log\rho_A)$;
	 	\item the conditional entropy $H(A|B)_\rho = H(AB)_\rho - H(B)_\rho$;
	 	\item the mutual information $I(A;B)_\rho = H(A)_\rho + H(B)_\rho - H(AB)_\rho = H(A)_\rho - H(A|B)_\rho$;
	 	\item the conditional mutual information $I(A;B|C)_\rho = H(A|C)_\rho - H(A|BC)_\rho$.
	 \end{itemize}
 	For classical states (i.e., probability distributions $(p(\clas{x}))_\clas{x}$) the von Neumann entropy reduces to the Shannon entropy $H(\clas{X}) = -\sum_{\clas{x}} p(\clas{x}) \log p(\clas{x})$.
 	In this paper all logarithms are taken to base $2$.
	For $n\in\mathbb{N}$ we use the notation $[n]=\lbrace 1,\dots,n\rbrace$.

\section{Upper Bounds on Key Rates}\label{sec:upper_bounds}

	\subsection{Rates of Device-independent Key Distribution Protocols}\label{sec:capacity-definitions}
	
		A general DIQKD protocol proceeds in two stages: data generation and classical processing. The ``standard'' protocol is given as Protocol~\ref{pro:diqkd} for reference.
		In the \emph{data generation stage} Alice and Bob use their devices  to generate the classical raw data (a sequence of bits) by playing $n$ non-local games, one after the other. A ``round'' $i\in[n]$ of the protocol corresponds to playing a single non-local game. 
		Some of the rounds are chosen at random to act as \emph{test rounds}, in which the parties use the device to play a non-local game such as the CHSH game.
		The other rounds are called \emph{key generation rounds}, where Alice and Bob choose predetermined inputs for the devices $\hat{\clas{x}}$ and $\hat{\clas{y}}$, respectively (see, e.g., Step~\ref{prostep:gen-in} in Protocol~\ref{pro:diqkd}). 
		In protocols based on the CHSH protocol, for example, Alice and Bob play the CHSH game with their devices $n\in\mathbb{N}$ times one after the other.\footnote{One may also consider protocols where all the games are played at once, in parallel; see~\cite[Chapter~6]{arnon2018reductions}. For the purpose of this work it does not matter if the games are played in sequence or in parallel.}
		
		In the \emph{classical processing stage} Alice and Bob use the raw data gathered in the generation stage.
		Firstly, they process the data to decide whether they should abort the protocol (in case they suspect that the device is not sufficiently good to produce a secure key) or not.
		The decision is based on a simple test: Alice and Bob count the fractions of the test rounds in which the non-local game is won; if it is above a chosen expected winning probability $\omega_{\textrm{exp}}$ they continue, otherwise they abort.
		This is the parameter estimation step~\ref{prostep:parm-est} in Protocol~\ref{pro:diqkd}. Among other things, the key rate will be a function of the chosen value for $\omega_{\textrm{exp}}$. 
		 
		Next, if Alice and Bob do not abort the protocol, they use the data produced by the devices in the key generation rounds, also called the raw key, to generate their final key. This is done by manipulating the raw key in various way and using classical communication between Alice and Bob. 
		We say that the protocol employs \emph{one-way communication} if the processing of the raw key is performed using classical communication only from Alice to Bob, i.e., Bob does not send any information about his raw data to Alice. The standard protocol is a one-way protocol  (see Step~\ref{prostep:one-way-proc} in Protocol~\ref{pro:diqkd}). 
		The most general protocol can include communication in both directions and is termed a \emph{two-way communication} protocol.

		\begin{algorithm}[H]
			\begin{algorithmic}[1]	
			
				\STATEx \textbf{Parameters:} 
				\STATEx $n$-- number of rounds
				\STATEx $\omega_{\textrm{exp}}$-- expected winning probability
				
				\STATEx

				\STATEx \textbf{Data generation stage:}
				
				\STATEx
				 
				\STATE  For every round $i\in[n]$ do Steps~\ref{prostep:input-qkd}-\ref{prostep:set-key}:
					
					\STATE\hspace{\algorithmicindent} Alice and Bob choose $\clas{t}_i\in\{0,1\}$ at random.\label{prostep:input-qkd}
					
					\STATE\hspace{\algorithmicindent} If  $\clas{t}_i=1$ they perform a \emph{test} round:
						\STATE\hspace{\algorithmicindent} \hspace{\algorithmicindent} Alice and Bob choose $\clas{x}_i\in\{0,1\},\clas{y}_i\in\{0,1\}$ at random.  
						\STATE\hspace{\algorithmicindent} \hspace{\algorithmicindent} They input $\clas{x}_i,\clas{y}_i$ to the device and record the outputs  $\clas{a}_i,\clas{b}_i$.  
				
					\STATE\hspace{\algorithmicindent} If  $\clas{t}_i=0$ they perform a \emph{key generation} round:
						\STATE\hspace{\algorithmicindent} \hspace{\algorithmicindent} Alice and Bob use a fixed generation input $\clas{x}_i = \hat{\clas{x}}=0,\clas{y}_i=\hat{\clas{y}}=2$.\label{prostep:gen-in}
						\STATE\hspace{\algorithmicindent} \hspace{\algorithmicindent} They input $\clas{x}_i,\clas{y}_i$ to the device and record the outputs  $\clas{a}_i,\clas{b}_i$.  \label{prostep:set-key}
				\STATEx
				
				\STATEx \textbf{Classical processing stage:}
				
				\STATEx
				
				\STATE \textbf{Parameter estimation:} Alice and Bob estimate the average winning probability in the CHSH game from the observed data in the \emph{test} rounds. If it is below the expected winning probability, $\omega_{\textrm{exp}}$, they abort.\label{prostep:parm-est}
				
				\STATEx
				
				\STATE \textbf{One-way classical post-processing:} Alice and Bob apply a one-way error correction protocol and a privacy amplification protocol (both classical) on their raw keys $\vec{\clas{A}}$ and $\vec{\clas{B}}$.\label{prostep:one-way-proc}
			
			\end{algorithmic}
		\caption{The ``standard'' device-independent quantum key distribution protocol}\label{pro:diqkd}
		\end{algorithm} 

		
		The \emph{key rate} of a QKD protocol is defined as the number of key bits, i.e., the length of the final key, divided by the number $n$ of rounds, or number of non-local games played during one execution of the protocol.
		We consider below the \emph{asymptotic IID key rate}, IID standing for ``independent and identically distributed''.
		This key rate refers to the rate of the protocol in the limit $n\rightarrow\infty$ and assuming that the devices are acting exactly the same in each round $i$, namely, in all rounds the correlation obtained from measuring the shared state is the same; see~\cite[Chapter~7]{arnon2018reductions} for further explanations.
		When looking for \emph{lower bounds} on the key rate of a protocol, considering the asymptotic IID setting is, at least a priori, restrictive. 
		In the context of the current work, however, this is not the case. 
		As we are looking for upper bounds, \emph{any} possible strategy of the adversary, wishing to gain information about the key produced by the protocol, sets an upper bound on the key rate.
		Indeed, one possibility for the adversary is to prepare a device for Alice and Bob that behaves identically and independently in each round.
		The IID assumption is common in information theory, as it describes a memoryless resource that acts in the same way in each communication round.

		To formally define key rates and capacities in this information-theoretic setting, let us get more technical.			
		Let  $\ket{\psi}_{ABE}$ denote the purification of $\rho_{AB}$, shared between Alice, Bob and the adversary Eve (this state is a qqq-state, meaning, all registers are quantum).
		To get Alice's classical outcome $\clas{A}$, her device measures $\rho_A$ using what we call the \emph{key measurement}: this is the measurement performed by the device in the key generation round, i.e., the one performed for the input  $\hat{\clas{x}}$. 
		In  Protocol~\ref{pro:diqkd}, for example, this is the measurement associated with the input $\clas{x}=0$ (Step~\ref{prostep:gen-in}). 
		The state $\ket{\psi}_{ABE}$ together with Alice's key measurement give us the cqq-state~$\rho_{\clas{A}BE}$, where the measured system $\clas{A}$ is classical, and $B$ and $E$ are quantum.
		
		For device-\emph{dependent} QKD protocols in terms of a cqq state $\rho_{\clas{A}BE}$ as defined above, the cqq-key capacity $K_{\mathrm{cqq}}(\rho_{\clas{A}BE})$ quantifies the amount of secure key that can be extracted from IID copies of $\rho_{\clas{A}BE}$.
		Informally, this capacity can be defined as follows (see App.~\ref{app:DD-capacities} for a formal definition):
		\begin{defn}[Device-dependent cqq-key capacity (informal)]\label{def:cqq-rate-informal}
			Given $n$ IID copies of a cqq state $\rho_{\clas{A}BE}$, the cqq-key capacity $K_{\mathrm{cqq}}(\rho_{\clas{A}BE})$ is defined as the largest key that can be extracted by Alice and Bob from $\rho_{\clas{A}BE}^{\otimes n}$ using local operations and public communication (LOPC) which allow Eve to learn only vanishing information about the established key (in the limit $n\to\infty$). \hfill$\diamond$
		\end{defn}
		
		It was shown by Devetak and Winter~\cite{devetak2005distillation} and by Renner and Christandl~\cite{renner2005simple} that this cqq-key capacity is bounded from below by 
		\begin{equation}
		K_{\mathrm{cqq}}(\rho_{\clas{A}BE}) \geq H(\clas{A}|E) - H(\clas{A}|B) \;. 
		\label{eq:DW-key-rate}
		\end{equation}
		This lower bound is achieved by protocols employing \emph{one-way communication} to perform error correction and privacy amplification in the classical processing step of the protocol. 
		Other more sophisticated protocols, e.g., protocols that use classical two-way communication, may improve on the lower bound given in Eq.~\eqref{eq:DW-key-rate}. 
		
		If Bob's key measurement is also fixed (in addition to Alice's measurement), then Alice and Bob's starting point is a ccq-state  $\rho_{\clas{A}\clas{B}E}$, derived by measuring the cqq-state $\rho_{\clas{A}BE}$ with Bob's key measurement associated to his input~$\hat{\clas{y}}$. 
		The corresponding key capacity of such a ccq-state can be defined informally as follows (see App.~\ref{app:DD-capacities} for a formal definition):
		\begin{defn}[Device-dependent ccq-key capacity (informal)]\label{def:ccq-rate-informal}
			Given $n$ IID copies of a ccq state $\rho_{\clas{A}\clas{B}E}$, the ccq-key capacity $K_{\mathrm{ccq}}(\rho_{\clas{A}\clas{B}E})$ is defined as the largest key that can be extracted by Alice and Bob from $\rho_{\clas{A}\clas{B}E}^{\otimes n}$ using local operations and public communication (LOPC) which allow Eve to learn only vanishing information about the established key (in the limit $n\to\infty$). \hfill$\diamond$
		\end{defn}
		
		We now turn our attention to device-\emph{independent} key capacities.
		We first introduce some terminology and notation. 
		All DIQKD protocols include a parameter estimation step in which Alice and Bob choose whether to abort or not. This is done according to the statistics that they observe, $p(\clas{a},\clas{b}|\clas{x},\clas{y})$, also termed a \emph{correlation}. 
		Consider now a tuple $\left( \sigma_{ABE}, \lbrace\Lambda^\clas{x}_\clas{a}\rbrace_\clas{a}, \lbrace \Lambda^\clas{y}_\clas{b} \rbrace_\clas{b}\right)$ comprising a qqq-state $\sigma_{ABE}$ and POVMs $\lbrace\Lambda^\clas{x}_\clas{a}\rbrace_\clas{a}, \lbrace \Lambda^\clas{y}_\clas{b} \rbrace_\clas{b}$, where it is understood that the tuple includes the POVMs for all $\clas{x}$ and $\clas{y}$.
		We define the set $\Sigma\left(p(\clas{a},\clas{b}|\clas{x},\clas{y})\right)$ of all such tuples compatible with a given correlation $p(\clas{a},\clas{b}|\clas{x},\clas{y})$:
		\begin{equation}
		\Sigma\left(p(\clas{a},\clas{b}|\clas{x},\clas{y})\right)	= \left\{ \left( \sigma_{ABE}, \lbrace\Lambda^\clas{x}_\clas{a}\rbrace_\clas{a}, \lbrace \Lambda^\clas{y}_\clas{b} \rbrace_\clas{b}\right) \; | \; 	\tr \left[ \left( \Lambda^\clas{x}_\clas{a} \otimes \Lambda^\clas{y}_\clas{b} \otimes \id_E \right) \sigma_{ABE} \right] = p(\clas{a},\clas{b}|\clas{x},\clas{y}) \right\} \;. 
		\label{eq:set-sigma}
		\end{equation}
		Abusing notation, we write that $\sigma \in \Sigma\left(p(\clas{a},\clas{b}|\clas{x},\clas{y})\right)$ if there exists a tuple with the state $\sigma$ (and suitable measurements). 
		Then,~$\sigma_{\clas{A}\clas{B}E}^{\hat{\clas{x}},\hat{\clas{y}}}$ is the state resulting from measuring $\sigma$ with the measurements in $\sigma$'s tuple that are associated with the predetermined ``key generation'' inputs $\hat{\clas{x}}\in\mathcal{X}$ and $\hat{\clas{y}}\in\mathcal{Y}$ (as in Step~\ref{prostep:gen-in} in Protocol~\ref{pro:diqkd}).
		
		In order to define DIQKD rates and capacities, we consider IID distributions $p(\clas{a},\clas{b}|\clas{x},\clas{y})^{\times n} = \prod_{i=1}^n p(\clas{a}_i,\clas{b}_i|\clas{x}_i,\clas{y}_i)$, which arise from Alice and Bob using the same correlation $p(\clas{a},\clas{b}|\clas{x},\clas{y})$ in each round of the protocol.
		The states and measurements giving rise to $p(\clas{a},\clas{b}|\clas{x},\clas{y})^{\times n}$ comprise the set $\Sigma\left(p(\clas{a},\clas{b}|\clas{x},\clas{y})^{\times n}\right)$.
		For a state $\sigma_{A^nB^nE^n}\in \Sigma\left(p(\clas{a},\clas{b}|\clas{x},\clas{y})^{\times n}\right)$ compatible with the correlation $p^{\times n}$, only the state $\sigma^{\hat{\clas{x}},\hat{\clas{y}}}_{\clas{A}^n\clas{B}^nE^n}$ obtained from using the key generation measurements $\hat{\clas{x}},\hat{\clas{y}}$ in each of the $n$ rounds is relevant for generating the key. 
		This state is given by
		\begin{align}
		\sigma^{\hat{\clas{x}},\hat{\clas{y}}}_{\clas{A}^n\clas{B}^nE^n} = \sum_{\clas{a}^n,\clas{b}^n} \proj{\clas{a}^n}_{\clas{A}^n} \otimes \proj{\clas{b}^n}_{\clas{B}^n} \otimes \tr_{\clas{A}^n\clas{B}^n} \left[ \left( \Pi^{\hat{\clas{x}}^n}_{\clas{a}^n} \otimes \Pi^{\hat{\clas{y}}^n}_{\clas{b}^n} \otimes \id_{E^n} \right) \sigma_{A^nB^nE^n} \right],
		\label{eq:ccq-state-protocol}
		\end{align}
		where $\hat{\clas{x}}^n = (\hat{\clas{x}},\dots,\hat{\clas{x}})$ and similarly for $\hat{\clas{y}}^n$.
		
		We are now ready to define the DI key capacity of a correlation. 
		
		\begin{defn}[Device-independent key capacity of a correlation]\label{def:DI_rate_corr}
			Let $p(\clas{a},\clas{b}|\clas{x},\clas{y})$ be a correlation and $\hat{\clas{x}}\in\cX, \hat{\clas{y}}\in\cY$ be a predetermined pair of inputs called the key generation inputs.
			\begin{enumerate}[(i)]
			\item\label{item:protocol} A \emph{key distillation protocol} consists of a sequence $\lbrace\Lambda_n\rbrace_{n\in\mathbb{N}}$ of maps $\Lambda_n\colon \clas{A}^n\clas{B}^n E^n\to \clas{K}_A\clas{K}_B E^n$ comprising local operations (with respect to the Alice-Bob-Eve partition) and public communication (LOPC), where $\clas{K}_A$ and $\clas{K}_B$ are the classical systems holding the final key of Alice and Bob.
			The goal of the protocol is for Alice and Bob to establish a classical maximally correlated state 
			\begin{align}
			\tau^L_{\clas{K}_A\clas{K}_B} = \frac{1}{2^{L}}\sum_{i=1}^{2^{L}} \proj{i}_{\clas{K}_A} \otimes \proj{i}_{\clas{K}_B}
			\end{align}
			so that Eve obtains only vanishing information about the key.
			A key distillation protocol is deemed secure in the device-independent setting if it produces a secure key for \emph{every} state $\sigma$ consistent with the observed correlation $p(\clas{a},\clas{b}|\clas{x},\clas{y})$.
			More formally, given $\eps>0$ and a key size $L_n$ for $n\in\mathbb{N}$, we require for all states $\sigma_{A^nB^nE^n}\in \Sigma\left(p(\clas{a},\clas{b}|\clas{x},\clas{y})^{\times n}\right)$ that
			\begin{align}
			\left\| \Lambda_n\left(\sigma^{\hat{\clas{x}},\hat{\clas{y}}}_{\clas{A}^n\clas{B}^nE^n} \right) - \tau^{L_n}_{\clas{K}_A\clas{K}_B} \otimes \rho_{E^n} \right\|_1 \leq \eps,
			\label{eq:security}
			\end{align}
			where $\sigma^{\hat{\clas{x}},\hat{\clas{y}}}_{\clas{A}^n\clas{B}^nE^n}$ is defined in \eqref{eq:ccq-state-protocol}, and $\rho_{E^n} = \tr_{\clas{K}_A\clas{K}_B} \Lambda_n\left(\sigma^{\hat{\clas{x}},\hat{\clas{y}}}_{\clas{A}^n\clas{B}^nE^n} \right)$.
		
			\item Given a sequence $\lbrace L_n\rbrace_{n\in\mathbb{N}}$, a rate $L=\limsup_{n\to\infty}(L_n/n)$ is called \emph{achievable} if \eqref{eq:security} holds for all $\eps>0$ and sufficiently large $n$.
			The \emph{device-independent key capacity} $K_{\mathrm{DI}}(p(\clas{a},\clas{b}|\clas{x},\clas{y}),\hat{\clas{x}},\hat{\clas{y}})$ of a correlation $p(\clas{a},\clas{b}|\clas{x},\clas{y})$ with key generation inputs $\hat{\clas{x}},\hat{\clas{y}}$ is defined as 
			\begin{align}
			K_{\mathrm{DI}}(p(\clas{a},\clas{b}|\clas{x},\clas{y}),\hat{\clas{x}},\hat{\clas{y}})=\sup\lbrace L\colon L \text{ is achievable.}\rbrace.
			\end{align}
		\end{enumerate}
		 \hfill$\diamond$
		\end{defn}
	
		An important aspect of Definition~\ref{def:DI_rate_corr} is that the security criterion \eqref{eq:security} for protocols as in Definition~\ref{def:DI_rate_corr}(\ref{item:protocol}) applies to \emph{all} states $\sigma_{A^nB^nE^n}\in \Sigma\left(p(\clas{a},\clas{b}|\clas{x},\clas{y})^{\times n}\right)$.
		In particular, it holds for IID extension $\sigma_{ABE}^{\otimes n}$ for any state $\sigma_{ABE}\in \Sigma\left(p(\clas{a},\clas{b}|\clas{x},\clas{y})\right)$.
		It then follows immediately that we have the following simple upper bound on the DI capacity:
		\begin{align}
		K_{\mathrm{DI}}(p(\clas{a},\clas{b}|\clas{x},\clas{y}),\hat{\clas{x}},\hat{\clas{y}}) \leq \inf_{\sigma\in\Sigma\left(p(\clas{a},\clas{b}|\clas{x},\clas{y})\right)} K_{\mathrm{ccq}}(\sigma^{\hat{\clas{x}},\hat{\clas{y}}}_{\clas{A}\clas{B}E}),
		\label{eq:ccq-upper-bound}
		\end{align}
		where $K_{\mathrm{ccq}}$ is the ccq-capacity from Definition~\ref{def:ccq-rate-informal}.		
		Note that Definition~\ref{def:DI_rate_corr} captures the fact that the conditions according to which a DI protocol aborts or not depend only on the observed correlation $p(\clas{a},\clas{b}|\clas{x},\clas{y})$ (assuming that errors in the classical post-processing steps vanish asymptotically). 
		
		When considering arbitrary states in the set $\Sigma\left(p(\clas{a},\clas{b}|\clas{x},\clas{y})^{\times n}\right)$ we cover protocols that can use the full knowledge of the correlation.
		In many cases, however, one considers protocols in which the aborting condition does not depend on the entire correlation, but only on partial information regarding it.  		
		The most common parameters relevant for deciding whether to abort the protocol or not are the expected winning probability $\omega_{\mathrm{exp}}$ in a certain non-local game and the expected quantum bit error rate 
		\begin{equation}\label{eq:qeber}
			Q_{\mathrm{exp}}=\Pr\left[ \clas{A}\neq\clas{B} | \clas{X}=\hat{\clas{x}},\clas{Y}=\hat{\clas{y}}\right] \;,
		\end{equation}
 		defined for the predetermined inputs $\hat{\clas{x}},\hat{\clas{y}}$.\footnote{Protocol~\ref{pro:diqkd} aborts if the observed average winning probability is below the expected winning probability $\omega_{\mathrm{exp}}$. Even though not explicitly written, the error correction protocol applied as part of the classical post-processing aborts with high probability when the observed error rate is above $Q_{\mathrm{exp}}$. Therefore, also in the standard protocol the two relevant parameters are $\omega_{\mathrm{exp}}$ and  $Q_{\mathrm{exp}}$.}
		For any correlation $p(\clas{a},\clas{b}|\clas{x},\clas{y})$, the winning probability $\omega$ in the relevant non-local game and the expected error rate $Q$	can be calculated from the correlation~$p(\clas{a},\clas{b}|\clas{x},\clas{y})$.
		That is, for any given $p(\clas{a},\clas{b}|\clas{x},\clas{y})$
		\begin{align}
			& \omega = \omega (p) \\
			& Q = Q (p)  \;.
		\end{align}
		
		For $n\in\mathbb{N}$ let $\hat{\Sigma}_n(\omega^{\star},Q^{\star})$  be the following set of tuples:
		\begin{equation}
			\hat{\Sigma}_n(\omega^{\star},Q^{\star}) = \bigcup_{\substack{p(\clas{a},\clas{b}|\clas{x},\clas{y})\colon\\ \; \omega(p)=\omega^{\star} \land Q(p)=Q^{\star}}} \Sigma\left( p(\clas{a},\clas{b}|\clas{x},\clas{y})^{\times n}\right)  
		\end{equation}
		We abbreviate $\hat{\Sigma}(\omega^{\star},Q^{\star}) \equiv \hat{\Sigma}_1(\omega^{\star},Q^{\star})$. 
		Similarly to what was done above, we write that $\sigma \in \hat{\Sigma}_n(\omega^{\star},Q^{\star})$ if there exists a tuple with the state~$\sigma$ and, as before in \eqref{eq:ccq-state-protocol}, $\sigma_{\clas{A}^n\clas{B}^nE^n}^{\hat{\clas{x}},\hat{\clas{y}}}$ is the state obtained from measuring $\sigma_{A^nB^nE^n}$ with the measurements in $\sigma$'s tuple that are associated with the inputs $\hat{\clas{x}}$ and $\hat{\clas{y}}$.
		
		We can then define the DI key capacity of $\omega^{\star}$ and $Q^{\star}$ as follows.
		\begin{defn}[Device-independent key-capacity of a winning probability]\label{def:DI_rate}
			For any $\omega^{\star}$ and $Q^{\star}$, the \emph{DI key capacity} $K_{\mathrm{DI}}(\omega^{\star},Q^{\star})$ of $\omega^{\star}$ and $Q^{\star}$ is defined exactly in the same way as in Definition~\ref{def:DI_rate_corr}, except that now the security criterion \eqref{eq:security} needs to hold for all $\sigma_{A^nB^nE^n}\in \hat{\Sigma}_n\left(\omega^{\star},Q^{\star}\right).$
			\hfill$\diamond$
		\end{defn}

		Clearly, for any $p(\clas{a},\clas{b}|\clas{x},\clas{y})$ with $\omega(p)=\omega^{\star}$ and $Q(p)=Q^{\star}$, 
		\begin{equation}
			K_{\mathrm{DI}}(p(\clas{a},\clas{b}|\clas{x},\clas{y}), \hat{\clas{x}},\hat{\clas{y}}) \geq K_{\mathrm{DI}}(\omega^{\star},Q^{\star}) \;,
		\end{equation}
		where $\hat{\clas{x}},\hat{\clas{y}}$ are related to $Q^{\star}$ via Equation~\eqref{eq:qeber}.
		Moreover, it holds that $\sigma_{ABE}^{\otimes n}\in \hat{\Sigma}_n(\omega^{\star},Q^{\star})$ for any $\sigma_{ABE}\in\hat{\Sigma}(\omega^{\star},Q^{\star})$, which gives us the following simple upper bound on $K_{\mathrm{DI}}(\omega^{\star},Q^{\star})$ in analogy with \eqref{eq:ccq-upper-bound} above:
		\begin{align}
		K_{\mathrm{DI}}(\omega^{\star},Q^{\star}) \leq \inf_{\sigma\in \hat{\Sigma}(\omega^{\star},Q^{\star})} K_{\mathrm{ccq}}(\sigma^{\hat{\clas{x}},\hat{\clas{y}}}_{\clas{A}\clas{B}E}).
			\label{eq:ccq-upper-bound-hat}
		\end{align}
		
		In summary, the following sequence of inequalities holds between the key capacities defined above:
		\begin{equation}\label{eq:capac_ineq}
		K_{\mathrm{DI}}(\omega^{\star},Q^{\star})   \leq 	K_{\mathrm{DI}}(p(\clas{a},\clas{b}|\clas{x},\clas{y}), \hat{\clas{x}},\hat{\clas{y}})  \leq	K_{\mathrm{ccq}}(\rho_{\clas{A}\clas{B}E}) \leq K_{\mathrm{cqq}}(\rho_{\clas{A}BE}) \;,
		\end{equation}
		for any $\rho_{ABE}\in \hat{\Sigma}(\omega^{\star},Q^{\star})$.
		
		In the above we defined the rates and capacities with respect to the generation inputs $\hat{\clas{x}},\hat{\clas{y}}$. In a recent work~\cite{Schwonnek2020robust} a new protocol with two generation inputs was suggested. Our definitions can be directly extended to include such protocols, and the same applies to the proofs given below.

		Note that the winning probability of a state in a non-local game, or the correlation that it can give rise to by measuring it, can be used to quantify the amount of entanglement of the state. 
		Defining the DI key capacities in terms of these quantities, as done above, paves the way to studying the usefulness of entanglement in the context of DIQKD protocols.


	\subsection{Intrinsic Information and Non-locality}\label{sec:intrinsic}
	
	In this section we discuss two information-theoretic quantities that serve as upper bounds on DI key capacities: the intrinsic information~\cite{christandl2007unifying} and the intrinsic non-locality~\cite{kaur2020fundamental}.

	\subsubsection{Intrinsic Information}
	
	We first define the \emph{intrinsic information} of a tripartite state, a quantum generalization of the classical intrinsic information that can be used to upper-bound the amount of secret key distillable from tripartite classical probability distributions \cite{maurer1997intrinsic}.
	Recall that for tripartite quantum states $\rho_{ABE}$ the \emph{quantum conditional mutual information} is defined as $I(A;B|E)_\rho = H(AE)_\rho + H(BE)_\rho - H(E)_\rho - H(ABE)_\rho$.
	\begin{defn}[Intrinsic information; \cite{christandl2007unifying}]\label{def:intrinsic-information}
		The \emph{intrinsic information} $I(A;B\downarrow E)_\rho$ of a tripartite quantum state $\rho_{ABE}$ is defined as
		\begin{align}\label{eq:intrinsic-information}
		I(A;B\downarrow E)_\rho = \inf_{\Lambda\colon E\to E'} I(A;B|E')_\sigma,
		\end{align}
		where the infimum is taken over quantum systems $E'$ and quantum channels $\Lambda\colon E\to E'$, and $\sigma_{ABE'} = (\idmap_{AB}\otimes \Lambda)(\rho_{ABE})$.
	\end{defn}
	
	Christandl \textit{et al.}~showed that the intrinsic information is an upper bound on the key capacity of a qqq-state $\rho_{ABE}$~\cite[Theorem~3.5]{christandl2007unifying}.
	Specializing this result to a ccq state $\rho_{\clas{A}\clas{B}E}$ and the corresponding key capacity from Definition~\ref{def:ccq-rate-informal}, we have the following lemma.	

	\begin{lemma}[\cite{christandl2007unifying}]\label{lem:ccq_int_up|_bound}
		For any \emph{ccq-state} $\rho_{\clas{A}\clas{B}E}$, 
		\begin{equation}
			K_{\mathrm{ccq}}(\rho_{\clas{A}\clas{B}E}) \leq I(\clas{A};\clas{B}\downarrow E)_\rho.
		\end{equation}
	\end{lemma}

	When evaluating the intrinsic information or conditional mutual information on a specific state $\sigma_{\clas{A}\clas{B}E}^{\clas{x},\clas{y}}$ in a ccq state $\sigma_{\clas{A}\clas{B}E}$, we use the following notation for better readability:
	\begin{align}
	I(\clas{A};\clas{B}\downarrow E)_{\sigma(\clas{x},\clas{y})} &\equiv I(\clas{A};\clas{B}\downarrow E)_{\sigma^{\clas{x},\clas{y}}} &
	I(\clas{A};\clas{B}| E)_{\sigma(\clas{x},\clas{y})} &\equiv I(\clas{A};\clas{B}| E)_{\sigma^{\clas{x},\clas{y}}}
	\end{align}
	
	Lemma~\ref{lem:ccq_int_up|_bound} translates to the following upper bounds on the device-independent key capacities defined in Sec.~\ref{sec:capacity-definitions}:
	
	\begin{lemma}\label{lem:DI_rate_int_info}
		For any correlation $p(\clas{a},\clas{b}|\clas{x},\clas{y})$ and $\hat{\clas{x}}\in\cX,\hat{\clas{y}}\in\cY$,
		\begin{align}
			K_{\mathrm{DI}}(p(\clas{a},\clas{b}|\clas{x},\clas{y}), \hat{\clas{x}},\hat{\clas{y}}) &\leq  \inf_{\sigma \in \Sigma\left(p(\clas{a},\clas{b}|\clas{x},\clas{y})\right)}  I(\clas{A};\clas{B}\downarrow E)_{\sigma(\hat{\clas{x}},\hat{\clas{y}})} \leq  \inf_{\sigma \in \Sigma\left(p(\clas{a},\clas{b}|\clas{x},\clas{y})\right)}  I(\clas{A};\clas{B}| E)_{\sigma(\hat{\clas{x}},\hat{\clas{y}})} \;.
		\intertext{Similarly,}
			K_{\mathrm{DI}}(\omega^{\star},Q^{\star}) &\leq  \inf_{\sigma \in \hat{\Sigma}(\omega^{\star},Q^{\star})}  I(\clas{A};\clas{B}\downarrow E)_{\sigma(\hat{\clas{x}},\hat{\clas{y}})} \leq \inf_{\sigma \in \hat{\Sigma}(\omega^{\star},Q^{\star})}  I(\clas{A};\clas{B}| E)_{\sigma(\hat{\clas{x}},\hat{\clas{y}})}   \;.
		\end{align}
	\end{lemma}
	\begin{proof}
		In both equations, the first inequality follows directly by using Lemma~\ref{lem:ccq_int_up|_bound} in combination with the capacity upper bounds \eqref{eq:ccq-upper-bound} and \eqref{eq:ccq-upper-bound-hat}, respectively. 
		The second inequality follows from picking the identity channel $\idmap_{E}\colon E\to E$ in the definition \eqref{eq:intrinsic-information} of the intrinsic information.
	\end{proof}
	
	We note here that the definition of $I(\clas{A};\clas{B}\downarrow E)_\sigma$ involves an optimization over quantum channels $\Lambda\colon E\to E'$ and quantum systems $E'$ for which there is no known efficient algorithm: there are no known cardinality bounds on the auxiliary system $E'$ for quantities based on the conditional mutual information such as the intrinsic information or the (related) squashed entanglement \cite{christandl2004squashed}.
	Moreover, for general states $\sigma_{AB}$ the dimension of the auxiliary system might have to be taken much larger than $|A||B|$ \cite{christandl2012entanglement}.
	However, picking the identity channel yields a simple, albeit potentially less tight, upper bound on the key capacities, which we will employ in Section~\ref{sec:CHSH_bounds}.

	\subsubsection{Intrinsic Non-locality}\label{sec:qinl}

	Recently, Kaur \textit{et al.}~\cite{kaur2020fundamental} introduced a new information-theoretic quantity termed  the ``quantum intrinsic non-locality'' (QINL). In the notation of~\cite{kaur2020fundamental}, the QINL is defined as follows.  
	\begin{defn}[Quantum intrinsic non-locality; \cite{kaur2020fundamental}]\label{def:quantum-intrinsic-nonlocality}
		For a correlation $p(\clas{a},\clas{b}|\clas{x},\clas{y})$, the \emph{quantum intrinsic non-locality} $N^Q(\clas{A}; \clas{B})_p$ is defined as
		\begin{align}\label{eq:q_int_nl}
		N^Q(\clas{A}; \clas{B})_p = \sup_{p(\clas{x},\clas{y})} \inf_{\sigma_{\clas{A}\clas{B}E}} \sum_{\clas{x},\clas{y}} p(\clas{x},\clas{y}) I(\clas{A};\clas{B}|E)_{\sigma(\clas{x},\clas{y})},
		\end{align}
		where $p(\clas{x},\clas{y})$ is a bipartite probability distribution, and the infimum is over states $\sigma_{\clas{A}\clas{B}E}$ of the form
		\begin{subequations}
		\begin{align}
		\sigma_{\clas{A}\clas{B}E} &= \sum_{\clas{x},\clas{y}} p(\clas{x},\clas{y}) \sigma_{\clas{A}\clas{B}E}^{\clas{x},\clas{y}}\\
		\text{with} \quad \sigma_{\clas{A}\clas{B}E}^{\clas{x},\clas{y}} &= \sum_{\clas{a},\clas{b}} \proj{\clas{a}}_\clas{A} \otimes \proj{\clas{b}}_\clas{B} \otimes \tr_{AB} \left[ \left( \Lambda^\clas{x}_\clas{a} \otimes \Lambda^\clas{y}_\clas{b} \otimes \id_E \right) \rho_{ABE} \right].\label{eq:rhoXY_ABE}
		\intertext{Here, $\lbrace\Lambda^\clas{x}_\clas{a}\rbrace_\clas{a}$ and $\lbrace \Lambda^\clas{y}_\clas{b} \rbrace_\clas{b}$ are POVMs for all $\clas{x},\clas{y}$, and $\rho_{ABE}$ is any tripartite state such that}
		\tr \left[ \left( \Lambda^\clas{x}_\clas{a} \otimes \Lambda^\clas{y}_\clas{b} \otimes \id_E \right) \rho_{ABE} \right] &= p(\clas{a},\clas{b}|\clas{x},\clas{y}).
		\end{align}
		\label{eq:qinl-extension}
		\end{subequations}
	\end{defn}
	
	The QINL can be used to upper-bound a DI key capacity defined in a slightly different way than the one we use here; see p.~31 and Theorem 22 in the arXiv version of~\cite{kaur2020fundamental}. 
	 A priori it seems like the definition from~\cite[Page 31]{kaur2020fundamental} may be more general than Definitions~\ref{def:DI_rate_corr} and~\ref{def:DI_rate} presented above. Roughly speaking, the main difference between the definitions is that in our work we always consider a ``special'' pair of inputs $\hat{\clas{x}},\hat{\clas{y}}$ which are used to generate the key (in contrast to other inputs that are used to test the device). 
	 In protocols such as Protocol~\ref{pro:diqkd}, the special predetermined inputs are used in order to potentially \emph{increase} the key rate of the protocol. It is therefore not clear if our definition is more restrictive than the one used in~\cite{kaur2020fundamental}. 
	 Indeed, the \emph{lower} bounds studied in the literature are derived from a protocol admissible in both our Definitions \ref{def:DI_rate_corr}, \ref{def:DI_rate} and the one in \cite{kaur2020fundamental}.
	 We leave this question unanswered and distinguish explicitly between the definitions when relevant below.
	 Furthermore, we discuss in Appendix~\ref{app:qinl} how our results transfer to the setting of \cite{kaur2020fundamental}.

	\subsection{Explicit Bounds for CHSH-based Protocols}\label{sec:CHSH_bounds}

	Our goal now is to derive \emph{explicit} upper bounds on the key rate of \emph{any} DIQKD protocol based on the CHSH game. 
	By saying that the protocol is based on the CHSH game we mean that the decision on whether to abort the protocol or not depends on the observed winning probability $\omega$ in the CHSH game. 
	For convenience, in this section we switch to work with the \emph{violation} $S\in[2,2\sqrt{2}]$ of the CHSH inequality instead of the winning probability. The two are related via $\omega=1/2 + S/8$. 
	 An additional parameter of the protocol is the observed quantum bit error rate $Q\in[0,0.5]$ (defined in \eqref{eq:qeber}).

	To get our upper bound we first use that, for all $S$ and $Q$ and for all $\rho_{ABE}\in \hat{\Sigma}(S,Q)$, 
	\begin{equation}\label{eq:simple_DI_bound}
		K_{\mathrm{DI}}(S,Q) \leq \inf_{\sigma \in \hat{\Sigma}(S,Q)}  I(\clas{A};\clas{B}| E)_{\sigma(\hat{\clas{x}},\hat{\clas{y}})}  \leq I(\clas{A};\clas{B}| E)_{\rho(\hat{\clas{x}},\hat{\clas{y}})} \;,
	\end{equation}
	which follows directly from  Lemma~\ref{lem:DI_rate_int_info}.	
	
	To derive an explicit upper bound one may choose any $\rho_{ABE}\in \hat{\Sigma}(S,Q)$, and the associated measurement in $\rho$'s tuple in $ \hat{\Sigma}(S,Q)$, and directly calculate $I(\clas{A};\clas{B}| E)_{\rho(\hat{\clas{x}},\hat{\clas{y}})}$. 
	Motivated by the work of Pironio \textit{et al.}~\cite{pironio2009device} dealing with \emph{lower bounds} for a CHSH-based DIQKD protocol, we make the following choice of state and measurements.\footnote{Recall that for DIQKD protocols based on binary test events (such as CHSH-based protocols), the analysis of lower bounds on key rates for collective attacks may w.l.o.g.~be restricted to two-qubit systems, as proved in \cite{pironio2009device}. However, we note that this result is not needed for our analysis of upper bounds as we may pick any state (IID or non-IID) compatible with the observed statistics (cf.~Sec.~\ref{sec:capacity-definitions} and \ref{sec:intrinsic}).}

	Alice and Bob's quantum state is chosen to be
	\begin{align}
		\tilde{\rho}_{AB} = \frac{1+C}{2} \proj{\Phi_+} + \frac{1-C}{2} \proj{\Phi_-},
	\label{eq:dephasing-choi}
	\end{align}
	where $C = \sqrt{(S/2)^2-1}$ and $\ket{\Phi_\pm} = (\ket{00} \pm \ket{11})/\sqrt{2}$.
	The measurements for the test rounds, associated with inputs $\clas{x},\clas{y}\in\{0,1\}$, are associated with the observables
	\begin{align} 
	A_0 &= \sigma_z & B_0 &= \frac{1}{\sqrt{1+C^2}}\sigma_z + \frac{C}{\sqrt{1+C^2}} \sigma_x \\
	A_1 &= \sigma_x & B_1 &= \frac{1}{\sqrt{1+C^2}}\sigma_z - \frac{C}{\sqrt{1+C^2}} \sigma_x,
	\end{align}
	where $\sigma_x,\sigma_z$ are the well-known Pauli $X$- and $Z$-matrices.
	For the key generation rounds associated with the inputs $\hat{\clas{x}}=0$ and $\hat{\clas{y}}=2$, Alice's device measures $A_0 = \sigma_z$, and Bob's device uses the $B_2=\sigma_z$ measurement with probability $1-2Q$, and a random bit with probability $2Q$.
	This choice ensures an observed quantum bit error rate $Q$.
	
	Consider a purification of the state $\tilde{\rho}_{AB}$ in \eqref{eq:dephasing-choi},
	\begin{align}
	\ket{\psi}_{ABE} = \sqrt{\frac{1+C}{2}} \ket{\Phi_+}_{AB}\ket{0}_E + \sqrt{\frac{1-C}{2}}\ket{\Phi_-}_{AB}\ket{1}_E,
	\label{eq:dephasing-choi-pure}
	\end{align}
	and let $\tilde{\rho}_{\clas{A}\clas{B}E}^{\hat{\clas{x}},\hat{\clas{y}}}$  be the effective state obtained in the key generation round, i.e., the state obtained from Alice's device measuring $A_0$ and Bob's device measuring $B_2$ on $\proj{\psi}_{ABE}$:
	\begin{align}
	\begin{aligned}
	\tilde{\rho}_{\clas{A}\clas{B}E}^{\hat{\clas{x}},\hat{\clas{y}}} &= \frac{1-Q}{2} \left(\proj{00}_{\clas{A}\clas{B}}\otimes \tilde{\rho}_E^+ + \proj{11}_{\clas{A}\clas{B}}\otimes \tilde{\rho}_E^-\right) + \frac{Q}{2}\left(\proj{01}_{\clas{A}\clas{B}}\otimes \tilde{\rho}_E^+ + \proj{10}_{\clas{A}\clas{B}}\otimes \tilde{\rho}_E^-\right),\\
	\tilde{\rho}_E^\pm &= \frac{1}{2}\begin{pmatrix} 1+C & \pm \sqrt{1-C^2}\\\pm \sqrt{1-C^2} & 1-C \end{pmatrix}.
	\end{aligned}
	\label{eq:ccq-state}
	\end{align}
	
	The above state and measurements are of special importance: they saturate the \emph{lower bound} on $K_{\mathrm{DI}}(S,Q)$ derived by Pironio \textit{et al.}~\cite{pironio2009device}.
	Using Eq.~\eqref{eq:simple_DI_bound}, we now use the same state and measurements to derive an \emph{upper bound} on $K_{\mathrm{DI}}(S,Q)$:
	\begin{align}
		K_{\mathrm{DI}}(S,Q) \leq I(\clas{A};\clas{B}|E)_{\tilde{\rho}(\hat{\clas{x}},\hat{\clas{y}})} \;,
	\label{eq:KD-upper-bound}
	\end{align}
	with $\tilde{\rho}_{\clas{A}\clas{B}E}^{\hat{\clas{x}},\hat{\clas{y}}}$ as defined in \eqref{eq:ccq-state}.
	The right-hand side of \eqref{eq:KD-upper-bound} can be expressed in analytical form, as stated in Theorem~\ref{thm:intrinsic3d} below.
	The resulting upper bound is plotted as a surface in Figure~\ref{fig:intrinsic3d} for $S\in[2,2\sqrt{2}]$ and $Q\in[0,1/2]$.
	\begin{thm}\label{thm:intrinsic3d}
		For any $S\in[2,2\sqrt{2}]$ and $Q\in[0,1/2]$, 
		\begin{equation}\label{eq:depo_exp}
		K_{\mathrm{DI}}(S,Q)   \leq 1 + h(a_{S,Q}) - h(Q) - h\left((1+\sqrt{(S/2)^2-1})/2\right) \;,
		\end{equation}
		where $a_{S,Q} = \frac{1}{2}\left(1+\sqrt{1+Q(1-Q)(S^2-8)}\right)$, and $h(x)=-x\log x-(1-x)\log(1-x)$ is the binary entropy.
	\end{thm}
	
	\begin{proof}
		Lemma~\ref{lem:DI_rate_int_info} gives an upper bound on the device-independent key capacity $K_{\mathrm{DI}}(S,Q)$ in terms of the intrinsic information of the state $\tilde{\rho}_{\clas{A}\clas{B}E}^{\hat{\clas{x}},\hat{\clas{y}}}$ defined in \eqref{eq:ccq-state}:
		\begin{align}
		K_{\mathrm{DI}}(S,Q) \leq I(\clas{A};\clas{B}|E)_{\tilde{\rho}(\hat{\clas{x}},\hat{\clas{y}})}.
		\end{align}
		A simple calculation then shows that $I(\clas{A};\clas{B}|E)_{\tilde{\rho}(\hat{\clas{x}},\hat{\clas{y}})}$ is equal to the right-hand side of \eqref{eq:depo_exp}.
	\end{proof}

	\begin{figure}
		\centering
		\includegraphics{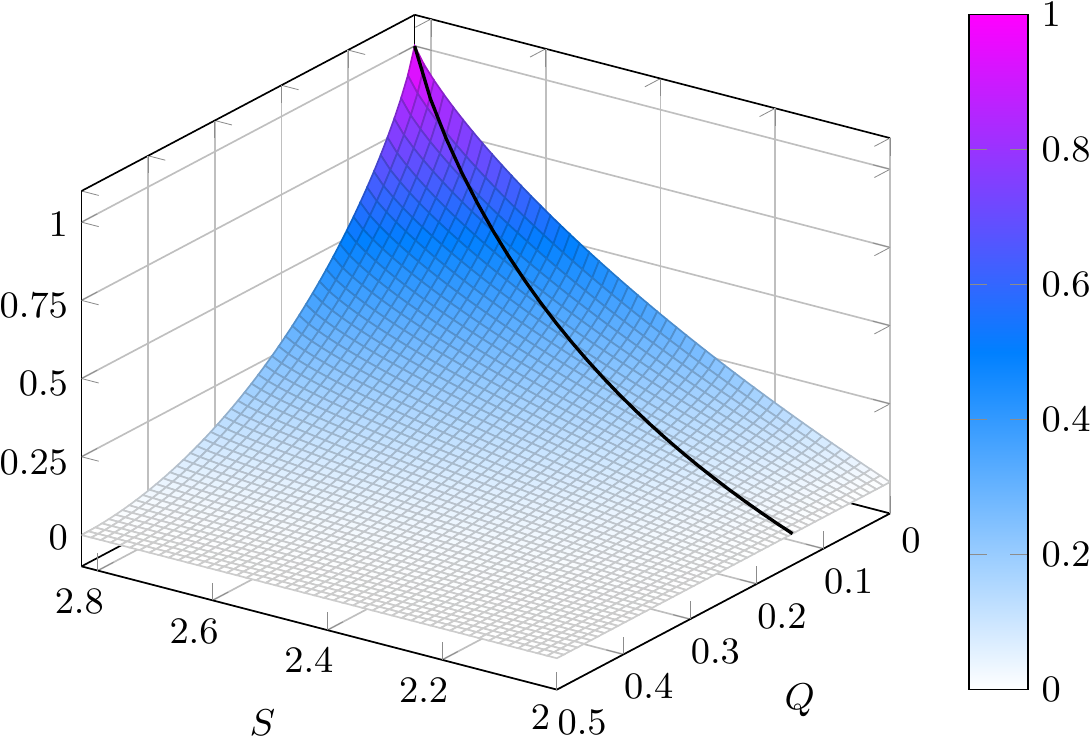}
		\caption{Surface plot of the intrinsic information $I(\clas{A};\clas{B}|E)_{\tilde{\rho}^{\hat{\clas{x}},\hat{\clas{y}}}}$ given in Theorem~\ref{thm:intrinsic3d} for $S\in[2,2\sqrt{2}]$ and $Q\in[0,1/2]$, where $\tilde{\rho}_{\clas{A}\clas{B}E}^{\hat{\clas{x}},\hat{\clas{y}}}$ is defined in \eqref{eq:ccq-state}.
			Note that $S$ and $Q$ are uncorrelated here.
		The black line is the bound given in Corollary~\ref{cor:intrinsic-fixedQ}, obtained from setting $Q=\frac{1}{2}(1-S/(2\sqrt{2}))$. }
		\label{fig:intrinsic3d}
	\end{figure}
	
	Any relation between $S$ and $Q$ can be chosen when constructing the protocol. To optimize the key rate, the relation between $S$ and $Q$ should fit the relation that Alice and Bob expect to see in the honest implementation of the protocol. 
	A common choice is to consider an honest implementation of the protocol in which maximally entangled states are being sent through a depolarizing channel with parameter $\nu\in[0,4/3]$.
	Alice and Bob then expect to hold a state
	\begin{equation}\label{eq:werner_state}
		\tau_{AB} = (1-\nu) \proj{\Phi_+} + \nu \id/4 \;,
	\end{equation} 	
	for which the relation $S=2\sqrt{2}(1-2Q)$ holds, with $Q=\nu/2$.
	This relation gives the well-known noise threshold of $7.1\%$; see~\cite[Figure~5]{arnon2018practical}. 
	We continue below with the relation $S=2\sqrt{2}(1-2Q)$ so we can compare our upper bound to that of~\cite{kaur2020fundamental}.
	In this situation, the bound in Theorem~\ref{thm:intrinsic3d} can be specialized to the following:
	\begin{cor}\label{cor:intrinsic-fixedQ}
		Let $S\in[2,\sqrt{2}]$ be given and set $Q=\frac{1}{2}(1-S/(2\sqrt{2}))$ (such as in the honest implementation of the CHSH protocol using a depolarizing channel considered above).
		Then, 
		\begin{align}
		K_{\mathrm{DI}}(S,Q)   \leq 1 + h(a_{S}) - h(Q) - h\left((1+\sqrt{(S/2)^2-1})/2\right) \;,\label{eq:KD-upper-bound-analytical}
		\end{align}
		where $a_S = \frac{1}{2}+\frac{1}{4\sqrt{2}}\sqrt{-32+16S^2-S^4}$. 
	\end{cor}
	We remark here that numerically optimizing the conditional mutual information $I(\clas{A};\clas{B}|E')_{\rho}$ over quantum channels $\Lambda\colon E\to E'$ for small $E'$ did not yield any improvement over \eqref{eq:KD-upper-bound-analytical}.
	
	The upper bound given in Eq.~\eqref{eq:KD-upper-bound-analytical} is plotted in dark blue/triangles in Fig.~\ref{fig:upper-bounds} as a function of $S$.
	Our upper bound is compared to the upper bound  derived in \cite[Eq.~(253)]{kaur2020fundamental} from the QINL $N^Q(\clas{A},\clas{B})_p$, for $p(\clas{a},\clas{b}|\clas{x},\clas{y})$ with violation $S$,\footnote{The QINL is a function of the entire correlation $p(\clas{a},\clas{b}|\clas{x},\clas{y})$, not of just $S$. However, the derivation of the bound in~\cite[Eq.~(253)]{kaur2020fundamental} depends only on $S$ and, de facto, identifies the special inputs $\hat{x},\hat{y}$ defining $Q$. Thus, the bounds can be easily compared.} appearing as the light blue/diamonds curve  in Fig.~\ref{fig:upper-bounds}. 
	Evidently, our upper bound in Eq.~\eqref{eq:KD-upper-bound} is tighter than the bound of~\cite{kaur2020fundamental} for any $S\in[2,2\sqrt{2}]$.
	
	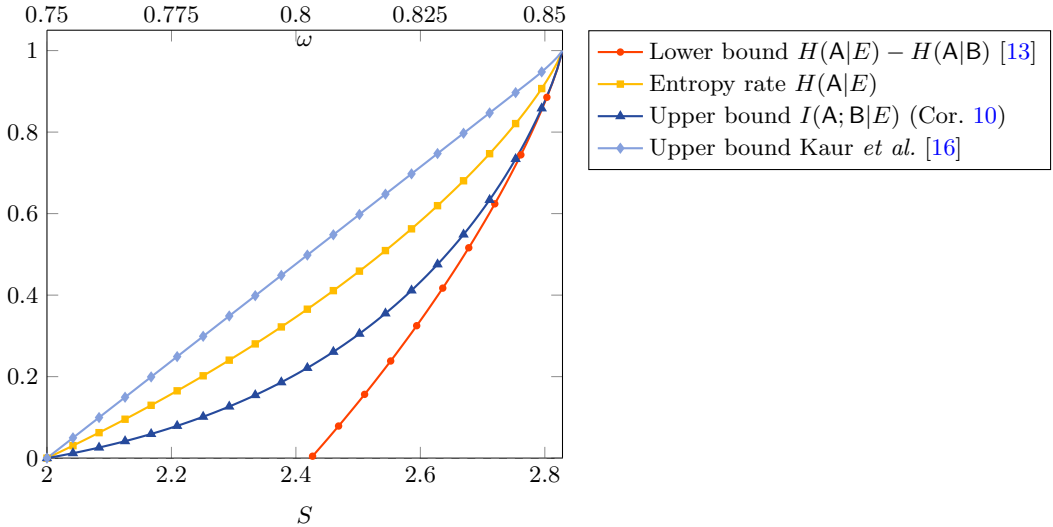
\begin{figure} 
		\centering
		\begin{tikzpicture}
		\begin{axis}[
		xlabel = $S$,
		xmin = 2,
		xmax = 2.8284,
		ymin = 0,
		ymax = 1.05,
		legend cell align = left,
		legend style = {at = {(1.05,1)}, anchor = north west},
		mark repeat = 5,
		axis x line*=bottom
		]
		\addplot[gray,dashed,domain=2:2.8284] {0}; 
		\addplot[thick,red,mark=*, mark size=1pt] table[x=S,y=DI] {pironio_intrinsic.dat};
		\addplot[thick,yellow,mark = square*, mark size = 1pt] table[x=S,y=DINOEC] {pironio_intrinsic.dat};
		\addplot[thick,dblue,mark = triangle*, mark size = 1.5pt] table[x=S,y=INT] {pironio_intrinsic.dat};
		\addplot[thick,mblue,mark = diamond*, mark size = 1.5pt] table[x=S,y=EWW] {pironio_intrinsic.dat}; 
		\legend{,Lower bound  $H(\clas{A}|E)-H(\clas{A}|\clas{B})$ \cite{pironio2009device}, Entropy rate $H(\clas{A}|E)$,Upper bound $I(\clas{A};\clas{B}|E)$ (Cor.~\ref{cor:intrinsic-fixedQ}), Upper bound Kaur \textit{et al.}~\cite{kaur2020fundamental}};
		\end{axis} 
		
		\begin{axis}
		[
		xlabel = $\omega$,
		x label style={at={(axis description cs:0.5,1.1)}},
		xmin = 2,
		xmax = 2.8284,
		xtick = {2,2.2,2.4,2.6,2.8},
		xticklabels = {$0.75$,$0.775$,$0.8$,$0.825$,$0.85$},
		hide y axis,
		axis x line* = top]
		\addplot[draw=none] {0};
		\end{axis}
		\end{tikzpicture}
		\caption{Comparison of achievable key rates and upper bounds plotted against CHSH violation $S$ (lower x-axis) and CHSH winning probability $\omega$ (upper x-axis). 
			The device-independent key rate evaluated in \cite{pironio2009device} is plotted in red/circles.
			The entropy rate is plotted in yellow/squares.
			The upper bound given by the dark blue line with triangle markers is the upper bound from Corollary~\ref{cor:intrinsic-fixedQ}, itself an upper bound on the intrinsic information $I(\clas{A};\clas{B}\downarrow E)_\rho$ from Definition~\ref{def:intrinsic-information}.
			All entropic quantities in the bounds above are evaluated on the state $\tilde{\rho}_{\clas{A}\clas{B}E}^{\hat{\clas{x}},\hat{\clas{y}}}$ defined in \eqref{eq:ccq-state}.
			The light blue/diamonds line is the upper bound derived in \cite{kaur2020fundamental}.}
		\label{fig:upper-bounds}
	\end{figure}

	Inspecting the proof method of~\cite{kaur2020fundamental} reveals a possible explanation of this: given the correlation $p$ with CHSH violation $S$, the bound in \cite[Eq.~(253)]{kaur2020fundamental} is obtained from evaluating the intrinsic information of $\tau_{AB}$ defined in Eq.~\eqref{eq:werner_state}, together with a convex decomposition argument similar to that used in \cite{goodenough2016repeaters}.
	In fact, the bound stays the same when neglecting the quantum register of the adversary, which corresponds to choosing the trivial quantum channel $\tr\colon E\to \mathbb{C}$ in \eqref{eq:intrinsic-information}.

	In contrast, we obtain our bound \eqref{eq:KD-upper-bound} by evaluating the intrinsic information of the state $\tilde{\rho}_{AB}$ in \eqref{eq:dephasing-choi} (resp.~$\tilde{\rho}^{\hat{\clas{x}},\hat{\clas{y}}}_{\clas{A}\clas{B}E}$ in \eqref{eq:ccq-state}), which Pironio \textit{et al.}~\cite{pironio2009device} showed to be Eve's optimal attack for a given CHSH violation $S$.
	Crucially, measuring $\tilde{\rho}_{AB}$ with the measurements chosen by Eve can be used to ``simulate'' the same correlations of the  state $\tau_{AB}$ used in~\cite{kaur2020fundamental}.
	Fig.~\ref{fig:upper-bounds} shows that the state $\tilde{\rho}^{\hat{\clas{x}},\hat{\clas{y}}}_{\clas{A}\clas{B}E}$ is not only optimal in terms of the (minimal) key distillation rate, but also a more judicious choice that decreases the intrinsic information, resulting in a tighter upper bound on the key rate.
		
	As a final remark let us note that, as mentioned before, Kaur \textit{et al.}~consider a potentially more general definition for a DI key capacity~\cite{kaur2020fundamental}. For DI protocols based on CHSH, our upper bound in Theorem~\ref{thm:intrinsic3d} can also be used directly to upper-bound the rates defined in~\cite{kaur2020fundamental} for a certain regime of the violation $S$.
	We refer to Appendix~\ref{app:qinl} for details.

\section{Revised Peres Conjecture}\label{sec:peres}

	\subsection{Bound Entanglement and Device-independent Cryptography}
		
		Entanglement is an essential resource in quantum cryptography. In the case of QKD, if Alice and Bob share a maximally entangled state then they can measure it to produce private and identical keys. 
		With the same idea in mind, entangled states from which maximally entangled states can be \emph{distilled} can also be used to produce a key. 
		Perhaps surprisingly, private keys can also be produced from \emph{bound entangled} states, namely, states whose entanglement cannot be distilled~\cite{horodecki1998mixed}---an observation that initiated the rich study of the so called ``private states''~\cite{horodecki2005secure,horodecki2009general}.
		A similar situation is encountered in the study of quantum channel capacities.\footnote{We note here that both the quantum and private capacity of a quantum channel are not enhanced by \emph{forward} classical communication from the sender to the receiver \cite{barnum2000quantum,bennett1996mixed,kretschmann2004tema,devetak2005private}. Hence, generating entanglement and secret key using a quantum channel and forward classical communication are the `dynamic' analogues of one-way entanglement distillation and one-way secret key distillation, respectively (see also \cite{devetak2005distillation}).}
		There are quantum channels which are useless for generating entanglement between the communicating parties, yet it is possible to generate a secret key using these channels \cite{horodecki2005secure,horodecki2008low}.
		This phenomenon lies at the heart of ``superactivation'' of quantum capacity, where two channels with zero quantum capacity can be used together to generate entanglement at a positive rate \cite{smith2008quantum}.
				
		A related question regarding the usefulness of bound entangled states was asked by Peres~\cite{peres1999all}. He conjectured that distillability is equivalent to Bell non-locality, i.e., bound entangled states cannot violate a Bell inequality (and hence cannot act as a resource for DI cryptography). 
		After an extensive search for bound entangled states and Bell inequalities that violate them, the conjecture was disproven by Vertesi and Brunner~\cite{vertesi2014disproving}. 
		It was further shown that the bound entangled state found in~\cite{vertesi2014disproving} can be used to produce randomness in a DI way, implying that it is a useful resource for DI randomness certification protocols.

		DIQKD protocols are more demanding than randomness certification protocols---not only should Alice have a random string, but Bob needs to hold the \emph{same} one. 
		The results of~\cite{vertesi2014disproving}, therefore, do not imply that bound entangled states can be used to produce a key using a DIQKD protocol.  
		We raise the following natural conjecture, a revision of the original Peres conjecture:
		\begin{conj}\label{conj:peres}
			Bound entangled states cannot be used to produce a key by any DIQKD protocol.
		\end{conj}
		 If true, the conjecture would imply a fundamental distinction between (device-\emph{dependent}) QKD and DIQKD---bound entangled states are a useful resource for the former but useless for the latter. 
  
  	\subsection{First Evidence}
  
		 As a first piece of evidence in favor of the conjecture, we show that the bound entangled state of~\cite{vertesi2014disproving} cannot be used to produce a key in a  DIQKD protocol of the form of Protocol~\ref{pro:diqkd}.
		 The bound entangled state of \cite{vertesi2014disproving} is given by
		 \begin{equation}\label{eq:be_state}
		 	\rho = \sum_{i=1}^4 \lambda_i \proj{\psi_i} \qquad\text{with}\qquad \lambda = \left( \frac{3257}{6884}, \frac{450}{1721},  \frac{450}{1721}, \frac{27}{6884} \right),
		 \end{equation}
		 and, for $a=\sqrt{\frac{131}{2}}$, the states $\ket{\psi_i}\in \mathbb{C}^3\times\mathbb{C}^3$ are
		 \begin{align}
		 	\ket{\psi_1} &=\frac{1}{\sqrt{2}}\left( \ket{00} + \ket{11}\right) \\
		 	\ket{\psi_2} &=\frac{a}{12}\left( \ket{01} + \ket{10}\right) + \frac{1}{60}\ket{02} - \frac{3}{10} \ket{21} \\
		 	\ket{\psi_3} &=\frac{a}{12}\left( \ket{00} - \ket{11}\right) + \frac{1}{60}\ket{12} + \frac{3}{10} \ket{20} \\
		 	\ket{\psi_4} &=\frac{1}{\sqrt{3}}\left( -\ket{01} + \ket{10}+\ket{22}\right) \;.
		 \end{align}
		 
		 The state $\rho$ can be used to violate a certain Bell inequality constructed in~\cite{vertesi2014disproving}. The relevant information in our context is that the inequality is such that Alice has three possible inputs $\clas{x}\in\{0,1,2\}$, with two outcomes for all measurements, namely, $\clas{a}\in\{0,1\}$. 
		 Bob has two possible measurements $\clas{y}\in\{0,1\}$. For $\clas{y}=0$ he has three possible outcomes and for $\clas{y}=1$ two outcomes; we write $\clas{b}\in\{0,1,2\}$ in both cases. 
		 
		 The state in Eq.~\eqref{eq:be_state} can be used to violate the inequality of~\cite[Eq.~(4)]{vertesi2014disproving}
		 when measured with the operators described as follows.
		 For Alice, let $q=1/5$ and
		 \begin{align}
		 	\ket{A_0} &= -q \ket{0} + \sqrt{3}q\ket{1} + \sqrt{1-4q^2} \ket{2} \\
		 	\ket{A_1} &= 2q\ket{0} + \sqrt{1-4q^2}\ket{2} \\
		 	\ket{A_2} &= -q \ket{0} - \sqrt{3}q\ket{1} + \sqrt{1-4q^2} \ket{2} \;.
		 \end{align}
		 Alice's measurement operators are then defined via $M_{\clas{a}=0|\clas{x}} = \proj{A_\clas{x}}$ and $M_{\clas{a}=1|\clas{x}} =  \id - M_{\clas{a}=0|\clas{x}}$. 
		 For Bob's first measurement, we set $M_{\clas{b}|0}=\proj{B_0^\clas{b}}$ for $\clas{b}\in\{0,1\}$ with
		 \begin{align}
		 	\ket{B_0^0} &= \sqrt{\frac{2}{3}}\ket{1} + \frac{1}{\sqrt{3}} \ket{2} \\
		 	\ket{B_0^1} &= - \frac{1}{\sqrt{2}}\ket{0} - \frac{1}{\sqrt{6}}\ket{1} + \frac{1}{\sqrt{3}}\ket{2} 
		 \end{align}
		 and $M_{2|0}= \id - M_{0|0} - M_{1|0}$. 
		 Bob's second measurement is given by $M_{0|1}=\proj{2}$ and $M_{1|1} = \id - M_{0|1}$.

		 As shown in~\cite[Table~1]{vertesi2014disproving}, measuring Bob's states with either of his measurements produces randomness that can be certified via the observed Bell violation. 
		 In other words, \emph{any} state and measurements leading to the violation achieved by the specific state and measurements above must lead to non-deterministic outcomes on Bob's side.
		 This means that the considered bound entangled state is a useful resource for DI randomness certification protocols.
		 
		Consider a modified version of Protocol~\ref{pro:diqkd}, where instead of playing the CHSH game (checking for the violation of the CHSH inequality) the used Bell inequality is that of~\cite[Eq.~(4)]{vertesi2014disproving}. 
		As the state and measurements violate the inequality, using them to execute the protocol will not cause the protocol to abort in the parameter estimation step with high probability. 
		Let us consider the case in which the communication is from Alice to Bob, and Alice's key measurement is $\clas{x}\in\{0,1,2\}$. 
		According to~\cite{devetak2005distillation} (recall also Eq.~\eqref{eq:DW-key-rate}), the entropy difference
		\begin{equation}\label{eq:rate_be_alice}
			H(\clas{A}|E) - H(\clas{A}|B) \;,
		\end{equation}
		evaluated on the cqq-state $\rho_{\clas{A}BE}^{\clas{x}}$ is a bound on the key rate of the protocol mentioned above. 
		We remark that this is an upper bound on the DI key rate of the protocol, which should be evaluated on the ccq-state  $\rho_{\clas{A}\clas{B}E}^{\clas{x},\hat{\clas{y}}}$, with $\hat{\clas{y}}=2$ a special key generation input not used for testing, similarly to the setting in Protocol~\ref{pro:diqkd}. 
		Using the data-processing inequality with respect to Bob's measurement then shows that $H(\clas{A}|E) - H(\clas{A}|B) \geq H(\clas{A}|E) - H(\clas{A}|\clas{B})$ in this situation.
		
		Similarly, one can consider a protocol in which the one-way communication is from Bob to Alice, Bob's key measurement is $\clas{y}\in\{0,1\}$ and Alice has a special key measurement $\hat{\clas{x}}=3$ not used for testing. 
		The relevant bound is then
		\begin{equation}\label{eq:rate_be_bob}
			H(\clas{B}|E) - H(\clas{B}|A) \;,
		\end{equation}
		evaluated on the qcq-state $\rho_{A\clas{B}E}^{\clas{y}}$.		
		In both cases, direct calculation shows that
		\begin{equation}
			\max_{\clas{x}} \; \{ H(\clas{A}|E) - H(\clas{A}|B) \} \leq 0 \qquad \text{and} \qquad \max_{\clas{y}} \; \{ H(\clas{B}|E) - H(\clas{B}|A)  \} \leq 0 \;.
		\end{equation}
		We learn from the above that the considered bound entangled state and measurements, which \emph{can} produce randomness, \emph{cannot} produce a key using a one-way DIQKD protocol.\footnote{For readers familiar with the work by Devetak and Winter~\cite{devetak2005distillation}, let us note that the statement above is not trivial as the measurements for Alice and Bob are not complete measurements (see the discussion in~\cite[Section~V]{devetak2005distillation}).} 
		This does not mean that other more complex protocols cannot use this state to produce a key in a DI way.
		Yet, as bound entangled states are a useful resource also for one-way device-\emph{dependent} QKD protocols \cite{horodecki2005secure,horodecki2008low}, considering one-way DIQKD as above is an interesting starting point.
				 
	\section{Discussion and Open Questions}
	
		In this work we used the intrinsic information to upper-bound the extractable key rates of a state using a DIQKD protocol. Our approach captures protocols in which a ``special'' generation input pair, $(\hat{x},\hat{y})$, is used to generate the key in the protocol. As shown in Section \ref{sec:CHSH_bounds}, our method yields tighter upper bounds compared to those derived in~\cite{kaur2020fundamental}. 
		The work can be directly extended to cover protocols in which several special generation inputs are used, e.g., the protocol suggested in~\cite{Schwonnek2020robust}.
		We note here once again that the QINL of~\cite{kaur2020fundamental} does not refer to special generation inputs  and hence may be more general. 
		A first interesting open question is whether protocols that do not use some fixed (set of) generation inputs can lead to higher key rates and, if so, in which regime of parameters. 
		
		To derive our explicit bounds we evaluated the intrinsic information on the state \emph{minimizing} the key rate of the standard protocol~\cite{pironio2009device}. It is intriguing to understand the relation between lower and upper bounds in this respect. Does the same state minimize both the lower \emph{and} the upper bound on the standard protocol key rate? 
		This question will be the subject of future investigation. 
		
		Presenting a violation of a Bell inequality is believed to be weaker than the ability to extract a DI key, i.e., some correlations may violate a Bell inequality but cannot be used to produce a key using a DIQKD protocol.
		Indeed, Conjecture~\ref{conj:peres} deals with one special case of this form.
		It is therefore interesting to find \emph{thresholds} below which no key can be created using any DIQKD protocol (meaning, below a certain violation or above a certain noise level it it impossible to extract a key).
		We want to point out a fundamental shortcoming of upper bounds derived from \emph{faithful} information-theoretic quantities.
		Here, a real-valued non-negative function $f$ evaluated on a correlation $p(\clas{a},\clas{b}|\clas{x},\clas{y})$ is called faithful if $f(p)=0$ holds if and only if $p$ is classical, i.e., it admits a local hidden-variable model.
		In the context of DIQKD, faithful measures are not desirable since upper bounds derived from them cannot lead to non-trivial \emph{thresholds} on the key rate-- a faithful measure must be strictly positive when evaluated on non-local correlations or states that violate a Bell inequality.
		Thus, even if there exists a certain threshold on the Bell violation (e.g., for the CHSH inequality) under which the correlation cannot be used to produce a key using \emph{any} DIQKD, this will not be revealed by a faithful measure.
		The QINL was shown to be faithful in \cite{kaur2020fundamental}; likewise, we show in Appendix~\ref{app:qinl} that the (worst-case) intrinsic information evaluated on the state and measurements from \cite{pironio2009device} (see Sec.~\ref{sec:CHSH_bounds}) upper-bounds the QINL, and hence inherits faithfulness from the latter.
		Finding quantities that are \emph{not} faithful and yet can be used to upper-bound the key rates of DIQKD protocols is an important open question, as such quantities may lead to improved thresholds on DIQKD.

	\paragraph*{Acknowledgments.}
	We thank Charles C.-W.~Lim for helpful feedback, and Mark M.~Wilde for valuable discussions regarding the definition of DI key capacities.
	The authors appreciate the hospitality of the University of Technology Sydney during the conference ``Beyond IID in Information Theory'', July 1-5, 2019, where this work was initiated.
	RAF acknowledges support from the Swiss National Science Foundation via the Postdoc.Mobility grant, the MURI Grant FA9550-18-1-0161 and ONR award N00014-17-1-3025.
	FL was supported by National Science Foundation (NSF) Grant No.~PHY 1734006 during part of this work, and acknowledges funding through the
	Army Research Lab CDQI program.

\bibliography{refs}

\begin{thebibliography}{37}%
\makeatletter
\providecommand \@ifxundefined [1]{%
 \@ifx{#1\undefined}
}%
\providecommand \@ifnum [1]{%
 \ifnum #1\expandafter \@firstoftwo
 \else \expandafter \@secondoftwo
 \fi
}%
\providecommand \@ifx [1]{%
 \ifx #1\expandafter \@firstoftwo
 \else \expandafter \@secondoftwo
 \fi
}%
\providecommand \natexlab [1]{#1}%
\providecommand \enquote  [1]{``#1''}%
\providecommand \bibnamefont  [1]{#1}%
\providecommand \bibfnamefont [1]{#1}%
\providecommand \citenamefont [1]{#1}%
\providecommand \href@noop [0]{\@secondoftwo}%
\providecommand \href [0]{\begingroup \@sanitize@url \@href}%
\providecommand \@href[1]{\@@startlink{#1}\@@href}%
\providecommand \@@href[1]{\endgroup#1\@@endlink}%
\providecommand \@sanitize@url [0]{\catcode `\\12\catcode `\$12\catcode
  `\&12\catcode `\#12\catcode `\^12\catcode `\_12\catcode `\%12\relax}%
\providecommand \@@startlink[1]{}%
\providecommand \@@endlink[0]{}%
\providecommand \url  [0]{\begingroup\@sanitize@url \@url }%
\providecommand \@url [1]{\endgroup\@href {#1}{\urlprefix }}%
\providecommand \urlprefix  [0]{URL }%
\providecommand \Eprint [0]{\href }%
\providecommand \doibase [0]{http://dx.doi.org/}%
\providecommand \selectlanguage [0]{\@gobble}%
\providecommand \bibinfo  [0]{\@secondoftwo}%
\providecommand \bibfield  [0]{\@secondoftwo}%
\providecommand \translation [1]{[#1]}%
\providecommand \BibitemOpen [0]{}%
\providecommand \bibitemStop [0]{}%
\providecommand \bibitemNoStop [0]{.\EOS\space}%
\providecommand \EOS [0]{\spacefactor3000\relax}%
\providecommand \BibitemShut  [1]{\csname bibitem#1\endcsname}%
\let\auto@bib@innerbib\@empty
\bibitem [{\citenamefont {Ekert}\ and\ \citenamefont
  {Renner}(2014)}]{ekert2014ultimate}%
  \BibitemOpen
  \bibfield  {author} {\bibinfo {author} {\bibfnamefont {A.}~\bibnamefont
  {Ekert}}\ and\ \bibinfo {author} {\bibfnamefont {R.}~\bibnamefont {Renner}},\
  }\href@noop {} {\bibfield  {journal} {\bibinfo  {journal} {Nature}\ }\textbf
  {\bibinfo {volume} {507}},\ \bibinfo {pages} {443} (\bibinfo {year}
  {2014})}\BibitemShut {NoStop}%
\bibitem [{\citenamefont {Brunner}\ \emph {et~al.}(2014)\citenamefont
  {Brunner}, \citenamefont {Cavalcanti}, \citenamefont {Pironio}, \citenamefont
  {Scarani},\ and\ \citenamefont {Wehner}}]{brunner2014bell}%
  \BibitemOpen
  \bibfield  {author} {\bibinfo {author} {\bibfnamefont {N.}~\bibnamefont
  {Brunner}}, \bibinfo {author} {\bibfnamefont {D.}~\bibnamefont {Cavalcanti}},
  \bibinfo {author} {\bibfnamefont {S.}~\bibnamefont {Pironio}}, \bibinfo
  {author} {\bibfnamefont {V.}~\bibnamefont {Scarani}}, \ and\ \bibinfo
  {author} {\bibfnamefont {S.}~\bibnamefont {Wehner}},\ }\href@noop {}
  {\bibfield  {journal} {\bibinfo  {journal} {Rev. Mod. Phys.}\ }\textbf
  {\bibinfo {volume} {86}},\ \bibinfo {pages} {419} (\bibinfo {year} {2014})},\
  \Eprint {http://arxiv.org/abs/1303.2849} {arXiv:1303.2849 [quant-ph]}
  \BibitemShut {NoStop}%
\bibitem [{\citenamefont {Scarani}(2019)}]{scarani2019bell}%
  \BibitemOpen
  \bibfield  {author} {\bibinfo {author} {\bibfnamefont {V.}~\bibnamefont
  {Scarani}},\ }\href@noop {} {\emph {\bibinfo {title} {Bell nonlocality}}}\
  (\bibinfo  {publisher} {Oxford University Press},\ \bibinfo {year}
  {2019})\BibitemShut {NoStop}%
\bibitem [{\citenamefont {Bell}(1964)}]{bell1964einstein}%
  \BibitemOpen
  \bibfield  {author} {\bibinfo {author} {\bibfnamefont {J.~S.}\ \bibnamefont
  {Bell}},\ }\href@noop {} {\bibfield  {journal} {\bibinfo  {journal}
  {Physics}\ }\textbf {\bibinfo {volume} {1}},\ \bibinfo {pages} {195}
  (\bibinfo {year} {1964})}\BibitemShut {NoStop}%
\bibitem [{\citenamefont {Clauser}\ \emph {et~al.}(1969)\citenamefont
  {Clauser}, \citenamefont {Horne}, \citenamefont {Shimony},\ and\
  \citenamefont {Holt}}]{clauser1969proposed}%
  \BibitemOpen
  \bibfield  {author} {\bibinfo {author} {\bibfnamefont {J.~F.}\ \bibnamefont
  {Clauser}}, \bibinfo {author} {\bibfnamefont {M.~A.}\ \bibnamefont {Horne}},
  \bibinfo {author} {\bibfnamefont {A.}~\bibnamefont {Shimony}}, \ and\
  \bibinfo {author} {\bibfnamefont {R.~A.}\ \bibnamefont {Holt}},\ }\href@noop
  {} {\bibfield  {journal} {\bibinfo  {journal} {Physical Review Letters}\
  }\textbf {\bibinfo {volume} {23}},\ \bibinfo {pages} {880} (\bibinfo {year}
  {1969})}\BibitemShut {NoStop}%
\bibitem [{\citenamefont {Ekert}(1991)}]{ekert1991quantum}%
  \BibitemOpen
  \bibfield  {author} {\bibinfo {author} {\bibfnamefont {A.~K.}\ \bibnamefont
  {Ekert}},\ }\href@noop {} {\bibfield  {journal} {\bibinfo  {journal}
  {Physical Review Letters}\ }\textbf {\bibinfo {volume} {67}},\ \bibinfo
  {pages} {661} (\bibinfo {year} {1991})}\BibitemShut {NoStop}%
\bibitem [{\citenamefont {Mayers}\ and\ \citenamefont
  {Yao}(1998)}]{mayers1998quantum}%
  \BibitemOpen
  \bibfield  {author} {\bibinfo {author} {\bibfnamefont {D.}~\bibnamefont
  {Mayers}}\ and\ \bibinfo {author} {\bibfnamefont {A.}~\bibnamefont {Yao}},\
  }in\ \href@noop {} {\emph {\bibinfo {booktitle} {Proceedings of the 39th
  Annual Symposium on Foundations of Computer Science}}}\ (\bibinfo
  {organization} {IEEE},\ \bibinfo {year} {1998})\ pp.\ \bibinfo {pages}
  {503--509}\BibitemShut {NoStop}%
\bibitem [{\citenamefont {Reichardt}\ \emph {et~al.}(2013)\citenamefont
  {Reichardt}, \citenamefont {Unger},\ and\ \citenamefont
  {Vazirani}}]{reichardt2013classical}%
  \BibitemOpen
  \bibfield  {author} {\bibinfo {author} {\bibfnamefont {B.~W.}\ \bibnamefont
  {Reichardt}}, \bibinfo {author} {\bibfnamefont {F.}~\bibnamefont {Unger}}, \
  and\ \bibinfo {author} {\bibfnamefont {U.}~\bibnamefont {Vazirani}},\
  }\href@noop {} {\bibfield  {journal} {\bibinfo  {journal} {Nature}\ }\textbf
  {\bibinfo {volume} {496}},\ \bibinfo {pages} {456} (\bibinfo {year}
  {2013})}\BibitemShut {NoStop}%
\bibitem [{\citenamefont {Vazirani}\ and\ \citenamefont
  {Vidick}(2014)}]{vazirani2014fully}%
  \BibitemOpen
  \bibfield  {author} {\bibinfo {author} {\bibfnamefont {U.}~\bibnamefont
  {Vazirani}}\ and\ \bibinfo {author} {\bibfnamefont {T.}~\bibnamefont
  {Vidick}},\ }\href@noop {} {\bibfield  {journal} {\bibinfo  {journal}
  {Physical Review Letters}\ }\textbf {\bibinfo {volume} {113}},\ \bibinfo
  {pages} {140501} (\bibinfo {year} {2014})},\ \Eprint
  {http://arxiv.org/abs/1210.1810} {arXiv:1210.1810 [quant-ph]} \BibitemShut
  {NoStop}%
\bibitem [{\citenamefont {Miller}\ and\ \citenamefont
  {Shi}(2016)}]{miller2016robust}%
  \BibitemOpen
  \bibfield  {author} {\bibinfo {author} {\bibfnamefont {C.~A.}\ \bibnamefont
  {Miller}}\ and\ \bibinfo {author} {\bibfnamefont {Y.}~\bibnamefont {Shi}},\
  }\href@noop {} {\bibfield  {journal} {\bibinfo  {journal} {Journal of the ACM
  (JACM)}\ }\textbf {\bibinfo {volume} {63}},\ \bibinfo {pages} {1} (\bibinfo
  {year} {2016})},\ \Eprint {http://arxiv.org/abs/1402.0489} {arXiv:1402.0489
  [quant-ph]} \BibitemShut {NoStop}%
\bibitem [{\citenamefont {Arnon-Friedman}\ \emph {et~al.}(2018)\citenamefont
  {Arnon-Friedman}, \citenamefont {Dupuis}, \citenamefont {Fawzi},
  \citenamefont {Renner},\ and\ \citenamefont {Vidick}}]{arnon2018practical}%
  \BibitemOpen
  \bibfield  {author} {\bibinfo {author} {\bibfnamefont {R.}~\bibnamefont
  {Arnon-Friedman}}, \bibinfo {author} {\bibfnamefont {F.}~\bibnamefont
  {Dupuis}}, \bibinfo {author} {\bibfnamefont {O.}~\bibnamefont {Fawzi}},
  \bibinfo {author} {\bibfnamefont {R.}~\bibnamefont {Renner}}, \ and\ \bibinfo
  {author} {\bibfnamefont {T.}~\bibnamefont {Vidick}},\ }\href@noop {}
  {\bibfield  {journal} {\bibinfo  {journal} {Nature communications}\ }\textbf
  {\bibinfo {volume} {9}},\ \bibinfo {pages} {459} (\bibinfo {year}
  {2018})}\BibitemShut {NoStop}%
\bibitem [{\citenamefont {Arnon-Friedman}\ \emph {et~al.}(2019)\citenamefont
  {Arnon-Friedman}, \citenamefont {Renner},\ and\ \citenamefont
  {Vidick}}]{arnon2019simple}%
  \BibitemOpen
  \bibfield  {author} {\bibinfo {author} {\bibfnamefont {R.}~\bibnamefont
  {Arnon-Friedman}}, \bibinfo {author} {\bibfnamefont {R.}~\bibnamefont
  {Renner}}, \ and\ \bibinfo {author} {\bibfnamefont {T.}~\bibnamefont
  {Vidick}},\ }\href@noop {} {\bibfield  {journal} {\bibinfo  {journal} {SIAM
  Journal on Computing}\ }\textbf {\bibinfo {volume} {48}},\ \bibinfo {pages}
  {181} (\bibinfo {year} {2019})},\ \Eprint {http://arxiv.org/abs/1607.01797}
  {arXiv:1607.01797 [quant-ph]} \BibitemShut {NoStop}%
\bibitem [{\citenamefont {Pironio}\ \emph {et~al.}(2009)\citenamefont
  {Pironio}, \citenamefont {Acin}, \citenamefont {Brunner}, \citenamefont
  {Gisin}, \citenamefont {Massar},\ and\ \citenamefont
  {Scarani}}]{pironio2009device}%
  \BibitemOpen
  \bibfield  {author} {\bibinfo {author} {\bibfnamefont {S.}~\bibnamefont
  {Pironio}}, \bibinfo {author} {\bibfnamefont {A.}~\bibnamefont {Acin}},
  \bibinfo {author} {\bibfnamefont {N.}~\bibnamefont {Brunner}}, \bibinfo
  {author} {\bibfnamefont {N.}~\bibnamefont {Gisin}}, \bibinfo {author}
  {\bibfnamefont {S.}~\bibnamefont {Massar}}, \ and\ \bibinfo {author}
  {\bibfnamefont {V.}~\bibnamefont {Scarani}},\ }\href@noop {} {\bibfield
  {journal} {\bibinfo  {journal} {New Journal of Physics}\ }\textbf {\bibinfo
  {volume} {11}},\ \bibinfo {pages} {045021} (\bibinfo {year} {2009})},\
  \Eprint {http://arxiv.org/abs/0903.4460} {arXiv:0903.4460 [quant-ph]}
  \BibitemShut {NoStop}%
\bibitem [{\citenamefont {Tan}\ \emph {et~al.}(2020)\citenamefont {Tan},
  \citenamefont {Lim},\ and\ \citenamefont {Renner}}]{tan2019advantage}%
  \BibitemOpen
  \bibfield  {author} {\bibinfo {author} {\bibfnamefont {E.~Y.-Z.}\
  \bibnamefont {Tan}}, \bibinfo {author} {\bibfnamefont {C.~C.-W.}\
  \bibnamefont {Lim}}, \ and\ \bibinfo {author} {\bibfnamefont
  {R.}~\bibnamefont {Renner}},\ }\href@noop {} {\bibfield  {journal} {\bibinfo
  {journal} {Physical Review Letters}\ }\textbf {\bibinfo {volume} {124}},\
  \bibinfo {pages} {020502} (\bibinfo {year} {2020})},\ \Eprint
  {http://arxiv.org/abs/1903.10535} {arXiv:1903.10535 [quant-ph]} \BibitemShut
  {NoStop}%
\bibitem [{\citenamefont {Schwonnek}\ \emph {et~al.}(2020)\citenamefont
  {Schwonnek}, \citenamefont {Goh}, \citenamefont {Primaatmaja}, \citenamefont
  {Tan}, \citenamefont {Wolf}, \citenamefont {Scarani},\ and\ \citenamefont
  {Lim}}]{Schwonnek2020robust}%
  \BibitemOpen
  \bibfield  {author} {\bibinfo {author} {\bibfnamefont {R.}~\bibnamefont
  {Schwonnek}}, \bibinfo {author} {\bibfnamefont {K.~T.}\ \bibnamefont {Goh}},
  \bibinfo {author} {\bibfnamefont {I.~W.}\ \bibnamefont {Primaatmaja}},
  \bibinfo {author} {\bibfnamefont {E.~Y.-Z.}\ \bibnamefont {Tan}}, \bibinfo
  {author} {\bibfnamefont {R.}~\bibnamefont {Wolf}}, \bibinfo {author}
  {\bibfnamefont {V.}~\bibnamefont {Scarani}}, \ and\ \bibinfo {author}
  {\bibfnamefont {C.~C.-W.}\ \bibnamefont {Lim}},\ }\href@noop {} {\bibfield
  {journal} {\bibinfo  {journal} {arXiv preprint}\ } (\bibinfo {year}
  {2020})},\ \Eprint {http://arxiv.org/abs/2005.02691} {arXiv:2005.02691
  [quant-ph]} \BibitemShut {NoStop}%
\bibitem [{\citenamefont {Kaur}\ \emph {et~al.}(2020)\citenamefont {Kaur},
  \citenamefont {Wilde},\ and\ \citenamefont {Winter}}]{kaur2020fundamental}%
  \BibitemOpen
  \bibfield  {author} {\bibinfo {author} {\bibfnamefont {E.}~\bibnamefont
  {Kaur}}, \bibinfo {author} {\bibfnamefont {M.~M.}\ \bibnamefont {Wilde}}, \
  and\ \bibinfo {author} {\bibfnamefont {A.}~\bibnamefont {Winter}},\
  }\href@noop {} {\bibfield  {journal} {\bibinfo  {journal} {New Journal of
  Physics}\ }\textbf {\bibinfo {volume} {22}},\ \bibinfo {pages} {023039}
  (\bibinfo {year} {2020})},\ \Eprint {http://arxiv.org/abs/1810.05627}
  {arXiv:1810.05627 [quant-ph]} \BibitemShut {NoStop}%
\bibitem [{\citenamefont {Winczewski}\ \emph {et~al.}(2019)\citenamefont
  {Winczewski}, \citenamefont {Das},\ and\ \citenamefont
  {Horodecki}}]{winczewski2019upper}%
  \BibitemOpen
  \bibfield  {author} {\bibinfo {author} {\bibfnamefont {M.}~\bibnamefont
  {Winczewski}}, \bibinfo {author} {\bibfnamefont {T.}~\bibnamefont {Das}}, \
  and\ \bibinfo {author} {\bibfnamefont {K.}~\bibnamefont {Horodecki}},\
  }\href@noop {} {\bibfield  {journal} {\bibinfo  {journal} {arXiv preprint}\ }
  (\bibinfo {year} {2019})},\ \Eprint {http://arxiv.org/abs/1903.12154}
  {arXiv:1903.12154 [quant-ph]} \BibitemShut {NoStop}%
\bibitem [{\citenamefont {Christandl}\ \emph {et~al.}(2007)\citenamefont
  {Christandl}, \citenamefont {Ekert}, \citenamefont {Horodecki}, \citenamefont
  {Horodecki}, \citenamefont {Oppenheim},\ and\ \citenamefont
  {Renner}}]{christandl2007unifying}%
  \BibitemOpen
  \bibfield  {author} {\bibinfo {author} {\bibfnamefont {M.}~\bibnamefont
  {Christandl}}, \bibinfo {author} {\bibfnamefont {A.}~\bibnamefont {Ekert}},
  \bibinfo {author} {\bibfnamefont {M.}~\bibnamefont {Horodecki}}, \bibinfo
  {author} {\bibfnamefont {P.}~\bibnamefont {Horodecki}}, \bibinfo {author}
  {\bibfnamefont {J.}~\bibnamefont {Oppenheim}}, \ and\ \bibinfo {author}
  {\bibfnamefont {R.}~\bibnamefont {Renner}},\ }in\ \href@noop {} {\emph
  {\bibinfo {booktitle} {Theory of Cryptography Conference}}}\ (\bibinfo
  {organization} {Springer},\ \bibinfo {year} {2007})\ pp.\ \bibinfo {pages}
  {456--478},\ \Eprint {http://arxiv.org/abs/quant-ph/0608199}
  {arXiv:quant-ph/0608199} \BibitemShut {NoStop}%
\bibitem [{\citenamefont {Peres}(1999)}]{peres1999all}%
  \BibitemOpen
  \bibfield  {author} {\bibinfo {author} {\bibfnamefont {A.}~\bibnamefont
  {Peres}},\ }\href@noop {} {\bibfield  {journal} {\bibinfo  {journal}
  {Foundations of Physics}\ }\textbf {\bibinfo {volume} {29}},\ \bibinfo
  {pages} {589} (\bibinfo {year} {1999})},\ \Eprint
  {http://arxiv.org/abs/quant-ph/9807017} {arXiv:quant-ph/9807017} \BibitemShut
  {NoStop}%
\bibitem [{\citenamefont {V{\'e}rtesi}\ and\ \citenamefont
  {Brunner}(2014)}]{vertesi2014disproving}%
  \BibitemOpen
  \bibfield  {author} {\bibinfo {author} {\bibfnamefont {T.}~\bibnamefont
  {V{\'e}rtesi}}\ and\ \bibinfo {author} {\bibfnamefont {N.}~\bibnamefont
  {Brunner}},\ }\href@noop {} {\bibfield  {journal} {\bibinfo  {journal}
  {Nature communications}\ }\textbf {\bibinfo {volume} {5}},\ \bibinfo {pages}
  {1} (\bibinfo {year} {2014})},\ \Eprint {http://arxiv.org/abs/1405.4502}
  {arXiv:1405.4502 [quant-ph]} \BibitemShut {NoStop}%
\bibitem [{\citenamefont {Moroder}\ \emph {et~al.}(2014)\citenamefont
  {Moroder}, \citenamefont {Gittsovich}, \citenamefont {Huber},\ and\
  \citenamefont {G{\"u}hne}}]{moroder2014steering}%
  \BibitemOpen
  \bibfield  {author} {\bibinfo {author} {\bibfnamefont {T.}~\bibnamefont
  {Moroder}}, \bibinfo {author} {\bibfnamefont {O.}~\bibnamefont {Gittsovich}},
  \bibinfo {author} {\bibfnamefont {M.}~\bibnamefont {Huber}}, \ and\ \bibinfo
  {author} {\bibfnamefont {O.}~\bibnamefont {G{\"u}hne}},\ }\href@noop {}
  {\bibfield  {journal} {\bibinfo  {journal} {Physical Review Letters}\
  }\textbf {\bibinfo {volume} {113}},\ \bibinfo {pages} {050404} (\bibinfo
  {year} {2014})},\ \Eprint {http://arxiv.org/abs/1405.0262} {arXiv:1405.0262
  [quant-ph]} \BibitemShut {NoStop}%
\bibitem [{\citenamefont {Arnon-Friedman}(2018)}]{arnon2018reductions}%
  \BibitemOpen
  \bibfield  {author} {\bibinfo {author} {\bibfnamefont {R.}~\bibnamefont
  {Arnon-Friedman}},\ }\emph {\bibinfo {title} {Reductions to IID in
  Device-independent Quantum Information Processing}},\ \href@noop {} {Ph.D.
  thesis},\ \bibinfo  {school} {ETH Z\"{u}rich} (\bibinfo {year} {2018}),\
  \Eprint {http://arxiv.org/abs/1812.10922} {arXiv:1812.10922 [quant-ph]}
  \BibitemShut {NoStop}%
\bibitem [{\citenamefont {Devetak}\ and\ \citenamefont
  {Winter}(2005)}]{devetak2005distillation}%
  \BibitemOpen
  \bibfield  {author} {\bibinfo {author} {\bibfnamefont {I.}~\bibnamefont
  {Devetak}}\ and\ \bibinfo {author} {\bibfnamefont {A.}~\bibnamefont
  {Winter}},\ }\href@noop {} {\bibfield  {journal} {\bibinfo  {journal}
  {Proceedings of the Royal Society A: Mathematical, Physical and engineering
  sciences}\ }\textbf {\bibinfo {volume} {461}},\ \bibinfo {pages} {207}
  (\bibinfo {year} {2005})},\ \Eprint {http://arxiv.org/abs/quant-ph/0306078}
  {arXiv:quant-ph/0306078} \BibitemShut {NoStop}%
\bibitem [{\citenamefont {Renner}\ and\ \citenamefont
  {Wolf}(2005)}]{renner2005simple}%
  \BibitemOpen
  \bibfield  {author} {\bibinfo {author} {\bibfnamefont {R.}~\bibnamefont
  {Renner}}\ and\ \bibinfo {author} {\bibfnamefont {S.}~\bibnamefont {Wolf}},\
  }in\ \href@noop {} {\emph {\bibinfo {booktitle} {Advances in
  cryptology-ASIACRYPT 2005}}}\ (\bibinfo  {publisher} {Springer},\ \bibinfo
  {year} {2005})\ pp.\ \bibinfo {pages} {199--216}\BibitemShut {NoStop}%
\bibitem [{\citenamefont {Maurer}\ and\ \citenamefont
  {Wolf}(1997)}]{maurer1997intrinsic}%
  \BibitemOpen
  \bibfield  {author} {\bibinfo {author} {\bibfnamefont {U.}~\bibnamefont
  {Maurer}}\ and\ \bibinfo {author} {\bibfnamefont {S.}~\bibnamefont {Wolf}},\
  }in\ \href@noop {} {\emph {\bibinfo {booktitle} {Proceedings of IEEE
  International Symposium on Information Theory}}}\ (\bibinfo {organization}
  {IEEE},\ \bibinfo {year} {1997})\ p.~\bibinfo {pages} {88}\BibitemShut
  {NoStop}%
\bibitem [{\citenamefont {Christandl}\ and\ \citenamefont
  {Winter}(2004)}]{christandl2004squashed}%
  \BibitemOpen
  \bibfield  {author} {\bibinfo {author} {\bibfnamefont {M.}~\bibnamefont
  {Christandl}}\ and\ \bibinfo {author} {\bibfnamefont {A.}~\bibnamefont
  {Winter}},\ }\href@noop {} {\bibfield  {journal} {\bibinfo  {journal}
  {Journal of Mathematical Physics}\ }\textbf {\bibinfo {volume} {45}},\
  \bibinfo {pages} {829} (\bibinfo {year} {2004})},\ \Eprint
  {http://arxiv.org/abs/quant-ph/0308088} {arXiv:quant-ph/0308088} \BibitemShut
  {NoStop}%
\bibitem [{\citenamefont {Christandl}\ \emph {et~al.}(2012)\citenamefont
  {Christandl}, \citenamefont {Schuch},\ and\ \citenamefont
  {Winter}}]{christandl2012entanglement}%
  \BibitemOpen
  \bibfield  {author} {\bibinfo {author} {\bibfnamefont {M.}~\bibnamefont
  {Christandl}}, \bibinfo {author} {\bibfnamefont {N.}~\bibnamefont {Schuch}},
  \ and\ \bibinfo {author} {\bibfnamefont {A.}~\bibnamefont {Winter}},\
  }\href@noop {} {\bibfield  {journal} {\bibinfo  {journal} {Communications in
  Mathematical Physics}\ }\textbf {\bibinfo {volume} {311}},\ \bibinfo {pages}
  {397} (\bibinfo {year} {2012})},\ \Eprint {http://arxiv.org/abs/0910.4151}
  {arXiv:0910.4151 [quant-ph]} \BibitemShut {NoStop}%
\bibitem [{\citenamefont {Goodenough}\ \emph {et~al.}(2016)\citenamefont
  {Goodenough}, \citenamefont {Elkouss},\ and\ \citenamefont
  {Wehner}}]{goodenough2016repeaters}%
  \BibitemOpen
  \bibfield  {author} {\bibinfo {author} {\bibfnamefont {K.}~\bibnamefont
  {Goodenough}}, \bibinfo {author} {\bibfnamefont {D.}~\bibnamefont {Elkouss}},
  \ and\ \bibinfo {author} {\bibfnamefont {S.}~\bibnamefont {Wehner}},\
  }\href@noop {} {\bibfield  {journal} {\bibinfo  {journal} {New Journal of
  Physics}\ }\textbf {\bibinfo {volume} {18}},\ \bibinfo {pages} {063005}
  (\bibinfo {year} {2016})},\ \Eprint {http://arxiv.org/abs/1511.08710}
  {arXiv:1511.08710 [quant-ph]} \BibitemShut {NoStop}%
\bibitem [{\citenamefont {Horodecki}\ \emph {et~al.}(1998)\citenamefont
  {Horodecki}, \citenamefont {Horodecki},\ and\ \citenamefont
  {Horodecki}}]{horodecki1998mixed}%
  \BibitemOpen
  \bibfield  {author} {\bibinfo {author} {\bibfnamefont {M.}~\bibnamefont
  {Horodecki}}, \bibinfo {author} {\bibfnamefont {P.}~\bibnamefont
  {Horodecki}}, \ and\ \bibinfo {author} {\bibfnamefont {R.}~\bibnamefont
  {Horodecki}},\ }\href@noop {} {\bibfield  {journal} {\bibinfo  {journal}
  {Physical Review Letters}\ }\textbf {\bibinfo {volume} {80}},\ \bibinfo
  {pages} {5239} (\bibinfo {year} {1998})},\ \Eprint
  {http://arxiv.org/abs/quant-ph/9801069} {arXiv:quant-ph/9801069} \BibitemShut
  {NoStop}%
\bibitem [{\citenamefont {Horodecki}\ \emph {et~al.}(2005)\citenamefont
  {Horodecki}, \citenamefont {Horodecki}, \citenamefont {Horodecki},\ and\
  \citenamefont {Oppenheim}}]{horodecki2005secure}%
  \BibitemOpen
  \bibfield  {author} {\bibinfo {author} {\bibfnamefont {K.}~\bibnamefont
  {Horodecki}}, \bibinfo {author} {\bibfnamefont {M.}~\bibnamefont
  {Horodecki}}, \bibinfo {author} {\bibfnamefont {P.}~\bibnamefont
  {Horodecki}}, \ and\ \bibinfo {author} {\bibfnamefont {J.}~\bibnamefont
  {Oppenheim}},\ }\href@noop {} {\bibfield  {journal} {\bibinfo  {journal}
  {Physical Review Letters}\ }\textbf {\bibinfo {volume} {94}},\ \bibinfo
  {pages} {160502} (\bibinfo {year} {2005})},\ \Eprint
  {http://arxiv.org/abs/quant-ph/0309110} {arXiv:quant-ph/0309110} \BibitemShut
  {NoStop}%
\bibitem [{\citenamefont {Horodecki}\ \emph {et~al.}(2009)\citenamefont
  {Horodecki}, \citenamefont {Horodecki}, \citenamefont {Horodecki},\ and\
  \citenamefont {Oppenheim}}]{horodecki2009general}%
  \BibitemOpen
  \bibfield  {author} {\bibinfo {author} {\bibfnamefont {K.}~\bibnamefont
  {Horodecki}}, \bibinfo {author} {\bibfnamefont {M.}~\bibnamefont
  {Horodecki}}, \bibinfo {author} {\bibfnamefont {P.}~\bibnamefont
  {Horodecki}}, \ and\ \bibinfo {author} {\bibfnamefont {J.}~\bibnamefont
  {Oppenheim}},\ }\href@noop {} {\bibfield  {journal} {\bibinfo  {journal}
  {IEEE Transactions on Information Theory}\ }\textbf {\bibinfo {volume}
  {55}},\ \bibinfo {pages} {1898} (\bibinfo {year} {2009})},\ \Eprint
  {http://arxiv.org/abs/quant-ph/0506189} {arXiv:quant-ph/0506189} \BibitemShut
  {NoStop}%
\bibitem [{\citenamefont {Barnum}\ \emph {et~al.}(2000)\citenamefont {Barnum},
  \citenamefont {Knill},\ and\ \citenamefont {Nielsen}}]{barnum2000quantum}%
  \BibitemOpen
  \bibfield  {author} {\bibinfo {author} {\bibfnamefont {H.}~\bibnamefont
  {Barnum}}, \bibinfo {author} {\bibfnamefont {E.}~\bibnamefont {Knill}}, \
  and\ \bibinfo {author} {\bibfnamefont {M.~A.}\ \bibnamefont {Nielsen}},\
  }\href@noop {} {\bibfield  {journal} {\bibinfo  {journal} {IEEE Transactions
  on Information Theory}\ }\textbf {\bibinfo {volume} {46}},\ \bibinfo {pages}
  {1317} (\bibinfo {year} {2000})},\ \Eprint
  {http://arxiv.org/abs/quant-ph/9809010} {arXiv:quant-ph/9809010} \BibitemShut
  {NoStop}%
\bibitem [{\citenamefont {Bennett}\ \emph {et~al.}(1996)\citenamefont
  {Bennett}, \citenamefont {DiVincenzo}, \citenamefont {Smolin},\ and\
  \citenamefont {Wootters}}]{bennett1996mixed}%
  \BibitemOpen
  \bibfield  {author} {\bibinfo {author} {\bibfnamefont {C.~H.}\ \bibnamefont
  {Bennett}}, \bibinfo {author} {\bibfnamefont {D.~P.}\ \bibnamefont
  {DiVincenzo}}, \bibinfo {author} {\bibfnamefont {J.~A.}\ \bibnamefont
  {Smolin}}, \ and\ \bibinfo {author} {\bibfnamefont {W.~K.}\ \bibnamefont
  {Wootters}},\ }\href@noop {} {\bibfield  {journal} {\bibinfo  {journal}
  {Physical Review A}\ }\textbf {\bibinfo {volume} {54}},\ \bibinfo {pages}
  {3824} (\bibinfo {year} {1996})},\ \Eprint
  {http://arxiv.org/abs/quant-ph/9604024} {arXiv:quant-ph/9604024} \BibitemShut
  {NoStop}%
\bibitem [{\citenamefont {Kretschmann}\ and\ \citenamefont
  {Werner}(2004)}]{kretschmann2004tema}%
  \BibitemOpen
  \bibfield  {author} {\bibinfo {author} {\bibfnamefont {D.}~\bibnamefont
  {Kretschmann}}\ and\ \bibinfo {author} {\bibfnamefont {R.~F.}\ \bibnamefont
  {Werner}},\ }\href@noop {} {\bibfield  {journal} {\bibinfo  {journal} {New
  Journal of Physics}\ }\textbf {\bibinfo {volume} {6}},\ \bibinfo {pages} {26}
  (\bibinfo {year} {2004})},\ \Eprint {http://arxiv.org/abs/quant-ph/0311037}
  {arXiv:quant-ph/0311037} \BibitemShut {NoStop}%
\bibitem [{\citenamefont {Devetak}(2005)}]{devetak2005private}%
  \BibitemOpen
  \bibfield  {author} {\bibinfo {author} {\bibfnamefont {I.}~\bibnamefont
  {Devetak}},\ }\href@noop {} {\bibfield  {journal} {\bibinfo  {journal} {IEEE
  Transactions on Information Theory}\ }\textbf {\bibinfo {volume} {51}},\
  \bibinfo {pages} {44} (\bibinfo {year} {2005})},\ \Eprint
  {http://arxiv.org/abs/quant-ph/0304127} {arXiv:quant-ph/0304127} \BibitemShut
  {NoStop}%
\bibitem [{\citenamefont {{Horodecki}}\ \emph {et~al.}(2008)\citenamefont
  {{Horodecki}}, \citenamefont {{Pankowski}}, \citenamefont {{Horodecki}},\
  and\ \citenamefont {{Horodecki}}}]{horodecki2008low}%
  \BibitemOpen
  \bibfield  {author} {\bibinfo {author} {\bibfnamefont {K.}~\bibnamefont
  {{Horodecki}}}, \bibinfo {author} {\bibfnamefont {L.}~\bibnamefont
  {{Pankowski}}}, \bibinfo {author} {\bibfnamefont {M.}~\bibnamefont
  {{Horodecki}}}, \ and\ \bibinfo {author} {\bibfnamefont {P.}~\bibnamefont
  {{Horodecki}}},\ }\href@noop {} {\bibfield  {journal} {\bibinfo  {journal}
  {IEEE Transactions on Information Theory}\ }\textbf {\bibinfo {volume}
  {54}},\ \bibinfo {pages} {2621} (\bibinfo {year} {2008})},\ \Eprint
  {http://arxiv.org/abs/quant-ph/0506203} {arXiv:quant-ph/0506203} \BibitemShut
  {NoStop}%
\bibitem [{\citenamefont {Smith}\ and\ \citenamefont
  {Yard}(2008)}]{smith2008quantum}%
  \BibitemOpen
  \bibfield  {author} {\bibinfo {author} {\bibfnamefont {G.}~\bibnamefont
  {Smith}}\ and\ \bibinfo {author} {\bibfnamefont {J.}~\bibnamefont {Yard}},\
  }\href@noop {} {\bibfield  {journal} {\bibinfo  {journal} {Science}\ }\textbf
  {\bibinfo {volume} {321}},\ \bibinfo {pages} {1812} (\bibinfo {year}
  {2008})},\ \Eprint {http://arxiv.org/abs/0807.4935} {arXiv:0807.4935
  [quant-ph]} \BibitemShut {NoStop}%
\end{thebibliography}%

\appendix

\section{Formal definitions of device-dependent key capacities} \label{app:DD-capacities}

The following definition of a cqq-key capacity is adapted from the definition of key capacity in~\cite[Definition~2.2]{christandl2007unifying} to the cqq case.
\begin{defn}\label{def:cqq-rate}
	Let $\rho_{\clas{A}BE}$ be cqq state.
	Given $n$ IID copies of $\rho_{\clas{A}BE}$, a key distillation protocol consists of a sequence $(\Lambda_n)_{n\in\mathbb{N}}$ of maps $\Lambda_n\colon \clas{A}^n\clas{B}^n E^n\to \clas{K}_A\clas{K}_B E^n$ consisting of local (with respect to the Alice-Bob-Eve partition) operations and public communication (LOPC), where $\clas{K}_A$ and $\clas{K}_B$ are the classical systems holding the final key of Alice and Bob.
	We require that (a) the map restricted to $\cH_\clas{A}^{\otimes n}$ is classical, i.e., given by a stochastic map, and (b) the restriction to $\cH_B^{\otimes n} \otimes \cH_E^{\otimes n}$ is completely positive and trace-preserving (CPTP).
	The goal of the protocol is to transform $\rho_{\clas{A}BE}^{\otimes n}$ into a state close to the following ccq-state with perfect key of length~$L$:
	\begin{align}
	\tau^L_{\clas{A}\clas{B}E} = \frac{1}{2^L}\sum_{i=1}^{2^L} \proj{i}_{\clas{A}} \otimes \proj{i}_{\clas{B}} \otimes \tau_E,
	\end{align}
	where $\tau_E$ is some arbitrary fixed state independent of $i$.
	We call a key rate $L$ \emph{achievable}, if there is a key distillation protocol $(\Lambda_n)_{n\in\mathbb{N}}$ and $(L_n)_{n\in\mathbb{N}}$ such that
	\begin{align}
	\text{(i)}\quad & \lim_{n\to\infty} \left\|\Lambda_n\left(\rho_{\clas{A}BE}^{\otimes n}\right) - \tau^{L_n}_{\clas{A}^n\clas{B}^nE^n}\right\|_1 = 0\\
	\text{(ii)}\quad & \limsup_{n\to\infty} \frac{L_n}{n} = L.
	\end{align}
	The \emph{key capacity} of a \emph{cqq-state} $\rho_{\clas{A}BE}$ is defined as $K_{\mathrm{cqq}}(\rho_{\clas{A}BE}) = \sup\lbrace L\in\mathbb{R}\colon L\text{ is achievable}\rbrace.$ \hfill$\diamond$
\end{defn}


We also specialize this definition to a ccq-state $\rho_{\clas{A}\clas{B}E}$ obtained from measuring Bob's system of the state $\rho_{\clas{A}BE}$:
\begin{defn}\label{def:ccq-rate}
	The \emph{key capacity} $K_{\mathrm{ccq}}(\rho_{\clas{A}\clas{B}E})$ of a \emph{ccq-state} $\rho_{\clas{A}\clas{B}E}$ is defined in the same way as in Definition~\ref{def:cqq-rate}, except that here the restriction of the LOPC maps $\Lambda_n$ for $n\in\mathbb{N}$ to $\clas{A}^n\clas{B}^n$ is taken to be classical, and the restriction to $E^n$ is CPTP. \hfill$\diamond$
\end{defn}

Note that, by definition, for any pair of states $\rho_{\clas{A}BE}$ and $\rho_{\clas{A}\clas{B}E}$ related via a measurement on Bob's marginal, 
\begin{equation}\label{eq:cqq_to_ccq_rate}
K_{\mathrm{ccq}}(\rho_{\clas{A}\clas{B}E}) \leq K_{\mathrm{cqq}}(\rho_{\clas{A}BE}) \;,
\end{equation}
since the set of admissible protocols in the ccq case is a subset of the cqq-admissible protocols.

\section{Theorem~\ref{thm:intrinsic3d} and alternate definitions of the DI key capacity} \label{app:qinl}

Our upper bound in Theorem~\ref{thm:intrinsic3d} is a bound on the DI key capacities defined in Definitions~\ref{def:DI_rate_corr} and \ref{def:DI_rate} in the CHSH scenario.
In this appendix we investigate the relationship between our setting and the alternative definition of DI key capacity on p.~31 in the arXiv version of \cite{kaur2020fundamental}.
In that paper, Kaur \textit{et al.}~introduce the quantum intrinsic non-locality $N^Q(\clas{A};\clas{B})_p$ (see Definition~\ref{def:quantum-intrinsic-nonlocality} in the main text) and show that it is an upper bound on their definition of DI key capacity.
We show in the following that for a certain range of the CHSH violation $S$ our bound from Theorem~\ref{thm:intrinsic3d} is also an upper bound on $N^Q(\clas{A};\clas{B})_p$, and hence also upper-bounds the DI key capacity as defined in \cite{kaur2020fundamental}.

To see this, note that for a given correlation $p(\clas{a},\clas{b}|\clas{x},\clas{y})$ the quantum intrinsic non-locality $N^Q(\clas{A},\clas{B})_p$ can be rewritten in our notation as:
\begin{align}
N^Q(\clas{A}; \clas{B})_p &=  \sup_{p(\clas{x},\clas{y})} 
\inf_{\sigma \in \Sigma\left(p(\clas{a},\clas{b}|\clas{x},\clas{y})\right)} 
\sum_{\clas{x},\clas{y}} p(\clas{x},\clas{y}) I(\clas{A};\clas{B}|E)_{\sigma(\clas{x},\clas{y})}.
\label{eq:nq}
\end{align}
Let now $p(\clas{a},\clas{b}|\clas{x},\clas{y})$ be the correlation used in a DI protocol to distill key from a CHSH violation $S$ and error rate $Q$.
Since $p(\clas{a},\clas{b}|\clas{x},\clas{y})$ can be obtained from the state $\tilde{\rho}_{AB}$ in \eqref{eq:dephasing-choi} together with the measurements given below \eqref{eq:dephasing-choi}, we have $\tilde{\rho}_{AB}\in \Sigma\left(p(\clas{a},\clas{b}|\clas{x},\clas{y})\right)$.
Hence, for a fixed distribution $p(\clas{x},\clas{y})$,
\begin{align}
\inf_{\sigma \in \Sigma\left(p(\clas{a},\clas{b}|\clas{x},\clas{y})\right)} 
\sum_{\clas{x},\clas{y}} p(\clas{x},\clas{y}) I(\clas{A};\clas{B}|E)_{\sigma(\clas{x},\clas{y})} \leq \sum_{\clas{x},\clas{y}} p(\clas{x},\clas{y}) I(\clas{A};\clas{B}|E)_{\tilde{\rho}(\clas{x},\clas{y})},
\label{eq:nq-bound}
\end{align}
where $\tilde{\rho}_{\clas{A}\clas{B}E}^{\clas{x},\clas{y}}$ is obtained from measuring the state in \eqref{eq:dephasing-choi-pure}.
For $S\gtrsim 2.59$ and $\hat{\clas{x}},\hat{\clas{y}}$ as defined above \eqref{eq:ccq-state}, we have 
\begin{align}
I(\clas{A};\clas{B}|E)_{\tilde{\rho}(\hat{\clas{x}},\hat{\clas{y}})} \geq I(\clas{A};\clas{B}|E)_{\tilde{\rho}(\clas{x},\clas{y})} \quad \text{for all $\clas{x},\clas{y}$.}
\end{align}
For $2\leq S \lesssim 2.59$ the maximal conditional mutual information of $\tilde{\rho}_{\clas{A}\clas{B}E}^{\clas{x},\clas{y}}$ is achieved by the measurements $(\clas{x},\clas{y})=(1,1)$.
Thus, using \eqref{eq:nq} and \eqref{eq:nq-bound}, we get the following upper bound on the quantum intrinsic non-locality:
\begin{align}
N^Q(\clas{A},\clas{B})_p \leq \sup_{p(\clas{x},\clas{y})} \sum_{\clas{x},\clas{y}} p(\clas{x},\clas{y}) I(\clas{A};\clas{B}|E)_{\tilde{\rho}(\clas{x},\clas{y})} = \begin{cases}
I(\clas{A};\clas{B}|E)_{\tilde{\rho}(\hat{\clas{x}},\hat{\clas{y}})} & \text{for $S\gtrsim 2.59$,}\\
I(\clas{A};\clas{B}|E)_{\tilde{\rho}(1,1)} & \text{for $2\leq S \lesssim 2.59$.}
\end{cases}
\label{eq:qinl-bounds}
\end{align}
In Fig.~\ref{fig:qinl-bounds} we compare this bound on $N^Q(\clas{A},\clas{B})_p$ to the one derived in \cite{kaur2020fundamental}; evidently, for $S\gtrsim 2.2$ the bound in \eqref{eq:qinl-bounds} is tighter.

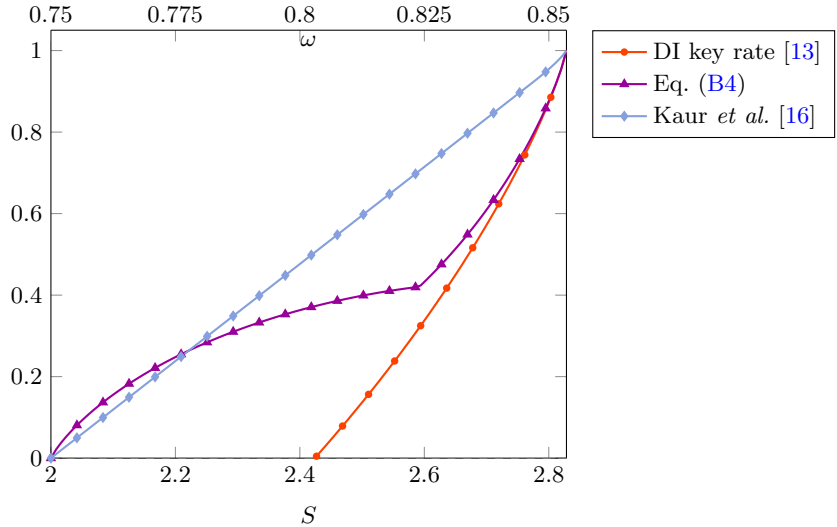
\begin{figure}
	\centering
	\begin{tikzpicture}
	\begin{axis}[
	xlabel = $S$,
	xmin = 2,
	xmax = 2.8284,
	ymin = 0,
	ymax = 1.05,
	legend cell align = left,
	legend style = {at = {(1.05,1)}, anchor = north west},
	mark repeat = 5,
	axis x line*=bottom
	]
	\addplot[gray,dashed,domain=2:2.8284] {0}; 
	\addplot[thick,red,mark=*, mark size=1pt] table[x=S,y=DI] {qinl.dat};
	\addplot[thick,purple,mark = triangle*, mark size = 1.5pt] table[x=S,y=QINL] {qinl.dat};
	\addplot[thick,mblue,mark = diamond*, mark size = 1.5pt] table[x=S,y=KWW] {qinl.dat}; 
	\legend{,DI key rate \cite{pironio2009device},Eq.~\eqref{eq:qinl-bounds}, Kaur \textit{et al.}~\cite{kaur2020fundamental}};
	\end{axis} 
	
	\begin{axis}
	[
	xlabel = $\omega$,
	x label style={at={(axis description cs:0.5,1.1)}},
	xmin = 2,
	xmax = 2.8284,
	xtick = {2,2.2,2.4,2.6,2.8},
	xticklabels = {$0.75$,$0.775$,$0.8$,$0.825$,$0.85$},
	hide y axis,
	axis x line* = top]
	\addplot[draw=none] {0};
	\end{axis}
	\end{tikzpicture}
	\caption{Comparison of upper bounds on the quantum intrinsic non-locality $N^Q(\clas{A},\clas{B})_p$ defined in Def.~\ref{def:quantum-intrinsic-nonlocality}, where $p(\clas{a},\clas{b}|\clas{x},\clas{y})$ is the correlation obtained from the CHSH game.
		Our upper bound \eqref{eq:qinl-bounds} on $N^Q(\clas{A},\clas{B})_p$ is plotted in purple/triangles, and the upper bound on $N^Q(\clas{A},\clas{B})_p$ from \cite{kaur2020fundamental} is plotted in light blue/diamonds.
		The device-independent key rate evaluated in \cite{pironio2009device} is plotted in red/circles.}
	\label{fig:qinl-bounds}
\end{figure}

\end{document}